\documentclass[lettersize,journal]{IEEEtran}
\usepackage{amsmath,amsfonts,amsthm}
\usepackage{mathrsfs}
\usepackage{multirow}
\usepackage{algorithmic}
\usepackage{algorithm}
\usepackage{array}
\UseRawInputEncoding
\usepackage{textcomp}
\usepackage{booktabs}
\usepackage{stfloats}
\usepackage{url}
\usepackage{verbatim}
\usepackage{graphicx}
\usepackage{cite}
\usepackage{varwidth}
\usepackage{subfig}
\usepackage{makecell}
\usepackage{bm}
\usepackage{booktabs}
\usepackage{balance}
\usepackage{setspace}
\usepackage{hyperref}
\usepackage{float}
\usepackage{caption}  % 在导言区引入 caption 宏包
\usepackage{mathtools}
\usepackage[dvipsnames]{xcolor}
\newtheorem{ppn}{Proposition}
\newtheorem{lemma}{Lemma}

\newcommand{\lmref}[1]{\textbf{Lemma \ref{#1}}}

\newcommand{\algref}[1]{\textbf{Algorithm \ref{#1}}}

\begin{document}
	
	\title{Time-Continuous Frequency Allocation \\ for Feeder Links of Mega Constellations \\ with Multi-Antenna Gateway Stations}
	
	\author{Zijun Liu, \textit{Student Member, IEEE}, Yafei Wang, \textit{Graduate Student Member, IEEE}, \\Tianhao Fang, \textit{Graduate Student Member, IEEE},  Wenjin Wang, \textit{Member, IEEE}, Zhili Sun, \textit{Senior Member, IEEE}
		\thanks{Manuscript received xxx.}
		\thanks{Zijun Liu, Yafei Wang, Tianhao Fang, and Wenjin Wang are with the National Mobile Communications Research Laboratory, Southeast University, Nanjing 210096, China, and also with Purple Mountain Laboratories, Nanjing 211100, China (E-mail: \{zijunliu, wangyf, tienhfang, wangwj\}@seu.edu.cn).}
		\thanks{Zhili Sun is with the Institute of Communication Systems, University of Surrey, GU2 7XH Guildford, U.K. (e-mail: z.sun@surrey.ac.uk).}
	}
	
	% The paper headers
	\markboth{}%
	{Shell \MakeLowercase{\textit{et al. }}: A Sample Article Using IEEEtran. cls for IEEE Journals}
	
	\maketitle
	
	\begin{abstract}
		With the {recent rapid advancement} of mega low earth orbit (LEO) satellite constellations, multi-antenna gateway station (MAGS) has emerged as a key enabler to support extremely high system capacity via massive feeder links. However, the densification of both space and ground segment {leads} to reduced spatial separation between links, posing unprecedented challenges of interference exacerbation. This paper investigates graph coloring-based frequency allocation methods for interference mitigation (IM) of mega LEO systems. We first reveal the characteristics of MAGS interference pattern and formulate the IM problem into a $K$-coloring problem using an adaptive threshold method. Then we propose two tailored graph coloring algorithms, namely Generalized Global (GG) and Clique-Based Tabu Search (CTS), to solve this problem. GG employs a low-complexity greedy conflict avoidance strategy, while CTS leverages the unique clique structure brought by MAGSs to enhance IM performance. Subsequently, we innovatively modify them to achieve time-continuous frequency allocation, which is crucial to ensure the stability of feeder links. Moreover, we further devise two mega constellation decomposition methods to alleviate the complexity burden of satellite operators. Finally, we propose a list coloring-based vacant subchannel utilization method to further improve spectrum efficiency and system capacity. Simulation results {on Starlink constellation of the first and second generations with 34396 satellites} demonstrate the effectiveness and superiority of the proposed methodology. 
	\end{abstract}
	
	\begin{IEEEkeywords}
		Graph theory, interference mitigation, multi-antenna gateway station, time-continuous frequency allocation, mega LEO satellite constellation. 
	\end{IEEEkeywords}
	
	\section{Introduction}\label{SecI}
	
	\IEEEPARstart{S}{ATELLITE} communication (SatCom) is among the key technologies of realizing the 6G vision of global seamless coverage, fundamentally revolutionizing the availability of broadband Internet \cite{6G,sats2,sat3}. Recently, low earth orbit (LEO) constellations are developing vigorously due to the reduction of satellite launch costs as well as  the advantages of low path loss and latency, providing high-quality pervasive connectivity even in remote regions lacking ground-based communications \cite{sats1}. 
	In order to satisfy the demand of high-speed data transmission, the scale of both space and ground segment of LEO systems will be unprecedentedly extended, enabling a substantial amount of feeder links to be established to offload tremendous data flow from space \cite{constellations}. Particularly, multi-antenna gateway station (MAGS) has emerged as a necessity at ground segment, trading off among the number of antennas required for supporting large system capacity, the limitations imposed by various geographical factors, and the financial feasibility of operation and maintenance \cite{gs1,gs2,starlinkiton,esstarlink}. However, this simultaneously aggravates the risk of harmful interference within and between LEO systems, causing severe throughput degradation or even service interruption \cite{worry,densesky}. Therefore, effective interference mitigation (IM) is an essential prerequisite for ensuring continuous high quality of service (QoS) \cite{dcss}.
	
	\subsection{Motivations}
	The ever-growing scale of LEO systems with MAGSs imposes three  significant challenges to designing IM methods. First, the dense deployment of MAGSs throughout the Earth inevitably leads to diminished spatial separation between feeder links, thereby severely exacerbating interference. Second, tremendous amount of feeder links significantly magnify the dimensions of IM problems, which results in increased time complexity. Finally, time-continuous resource allocation of feeder links is of utmost significance for both  maintaining the stability of SatCom systems and avoiding excessive signaling overhead, yet it has been scarcely considered when designing IM approaches. Although the highly dynamic LEO constellations results in a constantly varying interference environment and it becomes more turbulent in mega LEO systems, the resource allocation schemes are still required not to fluctuate dramatically during operation.
	
	Under such circumstances, it is unclear whether existing IM methods can still yield satisfactory performance under complexity and time continuity constraints, making it necessary to explore and utilize the interference characteristics of MAGSs to develop more advanced IM approaches.
	
	\subsection{Literature Review}
	Interference occurs when two signals in close spatial proximity occupy the same frequency spectrum simultaneously, i.e., they overlap in space, time and frequency domains concurrently. Therefore, IM can be achieved by isolating interfering sources in at least one of the three domains, which is no exception for SatCom systems. 
	
	Time domain IM for SatCom systems primarily includes beam hopping, which is usually optimized with space domain jointly in  user link resource allocation. At each time slot, the distance between active beams is maximized so that adjacent beams tend to illuminate in different time slots \cite{msbhlb,ngbh,bhmarl}. However, this technique is not applicable to feeder links due to high requirement of time continuity. 
	
	Space domain IM is applicable to both user and feeder links, in which ensuring sufficient angular separation is the main consideration \cite{newspace,worry,densesky,beampoint,MAGSIC,jptc,sbdp,wsy}. For feeder links, this can be achieved by altering satellite selection (look-aside) \cite{newspace,worry,densesky} or deviating beam pointing direction \cite{beampoint} when extremely small link angle occurs. \cite{MAGSIC} further optimized satellite selection with power allocation jointly for MAGSs. While these methods are effective for relatively small constellations, it is unclear whether they can adapt to mega LEO systems due to diminished spatial freedom caused by extremely high link density. Moreover, reducing interference by increasing angular distance also influences the quality of severing links. 
	
	Freqency domain IM methods mainly eliminate interference by separating interfering sources into subchannels using different frequency, which is also called channelization. For feeder links, previous works analyzed the effect of random subchannel allocation and identified it as a promising solution to enable constellation coexistence \cite{risk,newspace,worry,densesky}. \cite{aipredict} utilized a {neural network} trained on historical frequency usage data to predict and avoid co-frequency interference. By leveraging a more accurate integer linear programming (ILP)-based interference modeling, \cite{freqplan} proposed a high-performance generalized frequency plan method for user beams, and it was further optimized jointly with satellite routing problem in \cite{routingfassign}.
	While the complexity of ILP grows exponentially with problem dimensionality, graph coloring (GC)-based IM achieves excellent trade-off between {performance} and scalability, which has been widely investigated and deployed in various scenarios including SatCom  \cite{gcfirst,mmwave,topoif,mimo,uav, gcsatbeam}. By introducing the quantum annealing methodology, Hamiltonian Reduction is leveraged to further reduce its complexity \cite{qagc}. However, to the best of our knowledge, little effort has been devoted to the investigation of GC-based feeder link IM in mega LEO systems with MAGSs. Moreover, not only have the interference characteristics of MAGSs not been studied, but also the requirement of {time-continuous frequency allocation (TCFA)} for feeder links remains unconsidered.
	
	\vspace{-2mm}
	\subsection{Contributions}
	\vspace{-1mm}
	Inspired by the above discussions, in this paper we analyze the  interference characteristics of MAGSs and propose a series of IM methods based on GC to solve the dynamic IM problem of mega LEO systems. The main contributions and innovations of this work are summarized as follows.
	
	\begin{itemize}
		\item We first investigate the interference pattern of MAGSs  and reveal that mega LEO systems suffer from much stronger interference mainly due to diminished angular separation at both the satellite and the GS side, as well as increased number of interfering sources. Then we formulate the IM problem of mega LEO systems based on the regulation of ITU and transform it into a $K$-coloring problem using an adaptive threshold method, avoiding both excessive and insufficient interference protection.
		\item Two GC algorithms, Generalized Global (GG) and Clique-Based Tabu Search (CTS), are devised to solve the $K$-coloring problem, with GG prioritizing time efficiency by employing a greedy conflict avoidance strategy and CTS offering superior IM performance by leveraging the inherent clique structure brought by MAGSs. Based on the dynamic feature of the interference graph during operation, we propose time-continuous GG and CTS to maintain the IM performance under the requirement of TCFA. By exploiting the geographical distribution properties of interference, we further propose two mega constellation decomposition methods aiming at complexity reduction. 
		\item Based on system-level IM which ensures accessibility, we reveal that the proposed methodology can also yield outstanding system capacity performance. In order to further enhance spectrum efficiency, a list coloring-based vacant subchannel utilization method is devised to achieve more flexible frequency allocation without causing additional interference.
	\end{itemize}
	
	The rest of the paper is organized as the following. Section \ref{SecM} describes the system model and analyzes the interference characteristics of mega LEO systems. Section \ref{SecG} makes IM problem transformation based on GC. Section \ref{SecA} elaborates two time-continuous GC algorithms. Section \ref{SecR} delves into system capacity maximization. Section \ref{SecS} reports the simulation results. Section \ref{SecE} concludes this paper {and identifies directions for future research.}
	
	\begin{figure}[t]
		\centering
		\includegraphics[width=0.4\textwidth]{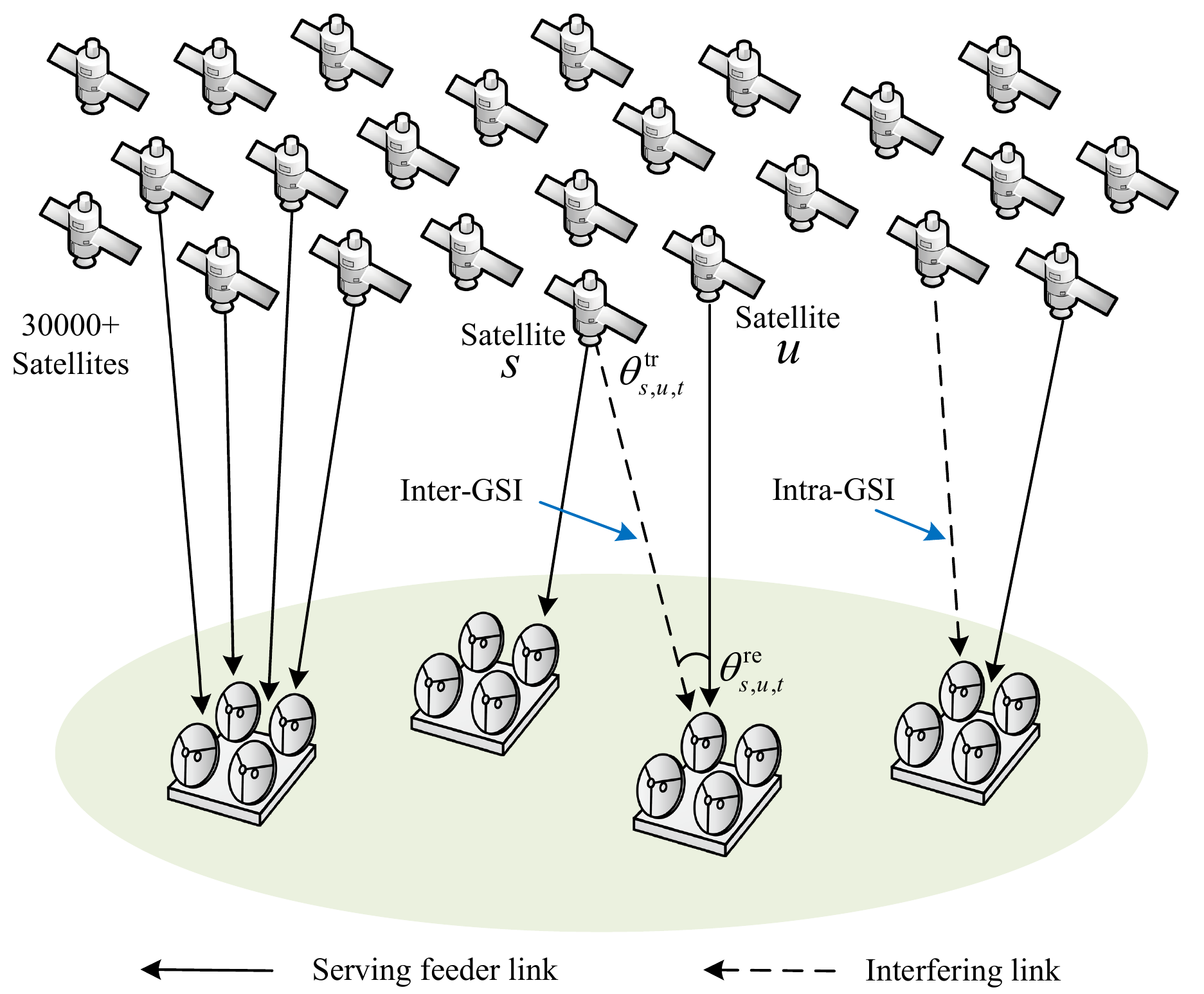}
		\caption{Interference scenario of a SatCom system with mega constellation and MAGSs.}
		\label{scene}
	\end{figure} 
	
	\section{System Model}\label{SecM}
	
	As shown in Fig. \ref{scene}, we consider the downlink scenario of a mega LEO system consisting of $S$ satellites and $Q$ multi-antenna gateway stations, where satellites transmit signals to MAGSs via feeder links simultaneously interfering others. The set of satellites and MAGSs are denoted as $\mathcal{S}=\{{1,2,...,S}\}$ and $\mathcal{Q}=\{{1,2,...,Q}\}$ respectively. 
	Each satellite is equipped with a single transmitting antenna of pattern S.1528 \cite{s1528} and each MAGS has $N_{\rm at}$ receiving antennas of pattern S.1428 \cite{s1428}, {as depicted in Fig. \ref{mavssa}}. The antenna gains of S.1528 and S.1428 versus off-axis angle $\phi$ are denoted as $G_{\rm s}(\phi)$ and $G_{\rm g}(\phi)$ respectively. 
	
	To reduce path loss and enhance link capacity, MAGSs employ maximum elevation angle principle for selecting communication satellites \cite{s1325}. MAGSs take turns selecting the currently available satellite with the largest elevation angle until all their antennas are occupied or there are no candidate satellites remaining. The set of selected working satellites at time slot $t$ and its cardinality are denoted as $\mathcal{W}(t)=\{s_1,s_2,...,s_{|\mathcal{W}(t)|}\}$ and $|\mathcal{W}(t)|$, respectively.
	
	{Channelization is employed for the system to realize controllable interference management \cite{risk,newspace,worry,densesky}, dividing the whole downlink feeder channel bandwidth $B_{\rm w}$ into $C$ subchannels with  equivalent bandwidth $B=B_{\rm w}/C$. The set of subchannel indices is represented as $\mathcal{C}=\{{1,2,...,C}\}$ and the index of subchannel assigned to satellite $s_i$ is denoted as $c_{s_i} \in \mathcal{C}$. For simplicity, we assume that each feeder link occupies exactly one complete subchannel.}
	
	\begin{figure*}[t]
		\centering
		\includegraphics[width=1\textwidth]{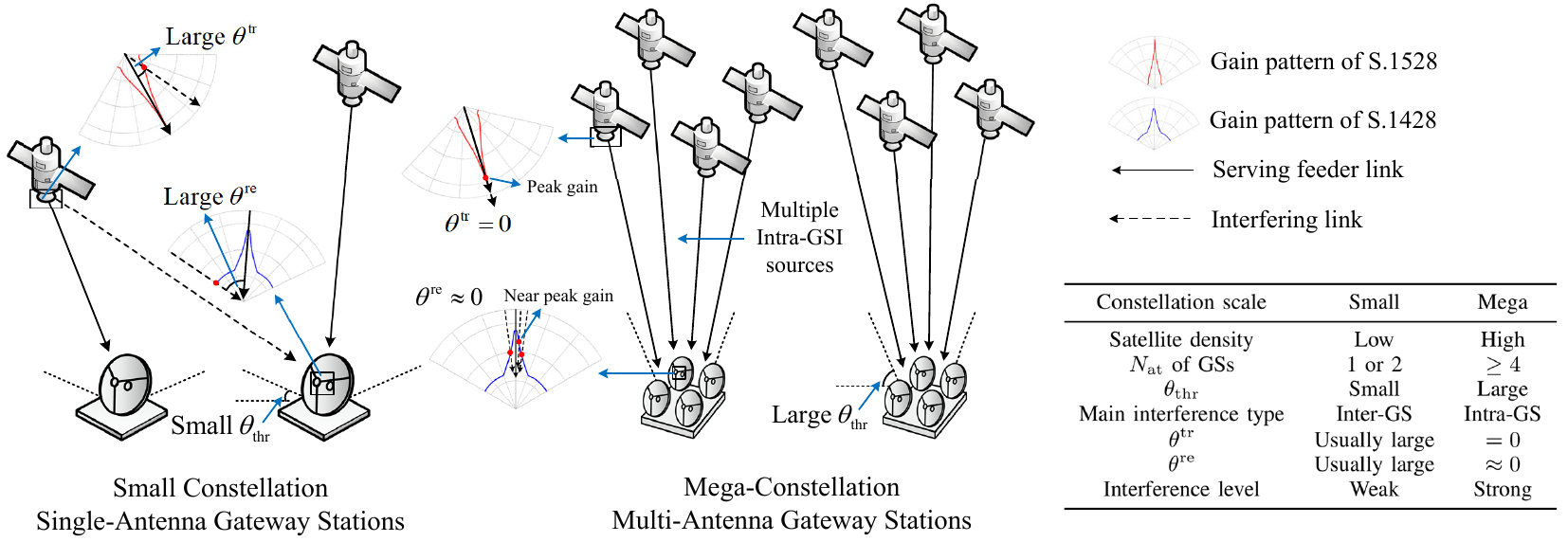}
		\caption{Interference comparison of MAGSs and SAGSs.}
		\vspace{-10pt}
		\label{mavssa}
	\end{figure*} 
	
	\subsection{Interference Model}
	
	Each feeder link between satellite $u$ and its receiving GS $q_u$ is interfered by all the other satellites above the horizontal plane of $q_u$ simultaneously communicating on subchannel $c_u$, forming the interfering satellite set $\mathcal{S}^I_{u}(c_u,t)$. As depicted in the middle part of Fig. \ref{scene}, for each interfering satellite $s \in \mathcal{S}^I_{u}(c_u,t)$ transmitting signal to GS $q_s$, the \textit{single-link interference} received by GS $q_u$ from satellite $s$ at time slot $t$ is calculated as
	\begin{equation}
		\label{I1}
		I_{s, u, t}=\frac{P_{s}^{\rm tr}(t)G_{\rm s}(\theta^{\rm tr}_{s, u,t})G_{\rm g}(\theta^{\rm re}_{s, u,t})}{(4\pi d_{s, u,t}/ \lambda_{s})^2},
	\end{equation}
	where $P_{s}^{\rm tr}(t)$ denotes the transmission power of satellite $s$, $\theta^{\rm tr}_{s, u,t}$ and $\theta^{\rm re}_{s, u,t}$ are off-axis angles of satellite $s$ and GS $q_u$ antennas respectively, $d_{s, u,t}$ represents the link distance from satellite $s$ to GS $q_u$, and $\lambda_{s}$ is the carrier wavelength. For simplicity in the following discussions, we classify single-link interference into two categories: intra-GS interference (Intra-GSI) if $q_s=q_u$ and inter-GS interference (Inter-GSI) otherwise ($q_s\neq q_u$). As a result, the \textit{aggregate interference} received by the corresponding receiving antenna of GS $q_u$  can be written as \cite{cofrin}
	\begin{equation}
		\label{sumI}
		\hat{I}_{u,t}=\sum_{s\in \mathcal{S}^I_{u}(c_u,t)} I_{s, u,t}.
	\end{equation} 
	We assume that receiving antennas have the same noise temperature $T_{\rm n}$, so that their noise power is all $N=\kappa T_{\rm n}B$, where $\kappa$ is the Boltzmann constant. Hence, the aggregate interference to noise ratio (I/N) of the feeder link between satellite $u$ and GS $q_u$ is defined as
	\begin{equation}
		(\hat{I}/N)_{u,t}=\frac{\hat{I}_{u,t}}{\kappa T_{\rm n}B}.
	\end{equation} 
	Similarly, single-link I/N for $I_{s, u,t}$ is defined as 
	\begin{equation}
		(I/N)_{s, u,t}=\frac{I_{s, u,t}}{\kappa T_{\rm n}B}.
	\end{equation} 
	
	In order to ensure the quality of feeder links, ITU regulates that the aggregate I/N of any feeder link should not exceed $I_{\rm th}^{\rm R}=-12.2$ dB \cite{s1432,densesky,beampoint,starlinkiton,mesinr}. Otherwise, it is considered as link failure (LF).   
	
	\subsection{Interference Analysis of Gateway Stations: Multi Antenna vs Single Antenna}\label{subsecIstr}
	
	Small scale SatCom systems mainly adopt single antenna gateway stations (SAGSs), which are typically spaced far apart and may connect to satellites with small elevation angles. As a result, sufficient spatial separation between feeder links is usually guaranteed, so that harmful interference rarely occurs. In comparison, mega LEO systems with dense space segment and MAGSs exhibit distinct interference patterns which severely exacerbate interference. Specifically, the main differences are as follows:
	
	\textit{1)} \textit{Smaller $\theta^{\rm tr}_{s, u}$:\footnote{For analyses irrelevant to time, the subscript $(\cdot)_t$ is omitted for simplicity.}} Different from SAGS, MAGS can access multiple satellites simultaneously. For all pairs of satellite $s$ and $u$ communicating with the same MAGS, they point their antennas directly towards the same location, resulting in $\theta^{\rm tr}_{s, u}=0$.
	
	\textit{2)} \textit{Smaller $\theta^{\rm re}_{s, u}$:} With the densification of mega constellations, MAGSs will have sufficient candidate satellites to select even with very large elevation angle threshold $\theta_{\rm thr}$. However, according to the maximum elevation angle principle, satellites selected by the same MAGS are very close to each other, which makes $\theta^{\rm re}_{s, u}$ extremely small. 
	
	\textit{3)} \textit{Multiple interfering sources:} Unlike SAGSs which only experience Inter-GSI, each link of MAGS encounters both Inter-GSI and strong Intra-GSI from the other $N_{\rm at}-1$ satellites communicating with the same MAGS.
	
	\textit{4)} \textit{Much stronger interference:} The three aforementioned characteristics of MAGS jointly lead to serious deterioration of interference, as summarized in Fig. \ref{mavssa}. On one hand, the gain of both S.1528 and S.1428 antennas increases sharply when the off-axis angle is close to zero. According to equation (\ref{I1}), smaller $\theta^{\rm tr}_{s, u}$ and $\theta^{\rm re}_{s, u}$ result in extremely strong Intra-GSI, which can even reach the level of signal  power. On the other hand, multiple interfering sources will further degrade the aggregate I/N by approximately $10\log_{10}(N_{\rm at}-1)$ dB compared to the single-link I/N. For example, this value is 8.45 dB for $N_{\rm at}=8$ and 13.8 dB for $N_{\rm at}=25$.
	
	Furthermore, {our simulation results show} the intensity of interference in mega LEO systems without channelization in Fig. \ref{i/nex}, where the complementary cumulative distribution functions (CCDF) of aggregate I/N for all links versus $N_{\rm at}$ are shown. The system parameters used in this simulation are provided in Table \ref{paras}. It is evident that the LF rate soars to 96\% even for $N_{\rm at}=2$, disrupting the normal operation of the system. 
	
	\begin{figure}[t]
		\centering
		\includegraphics[width=0.75\linewidth]{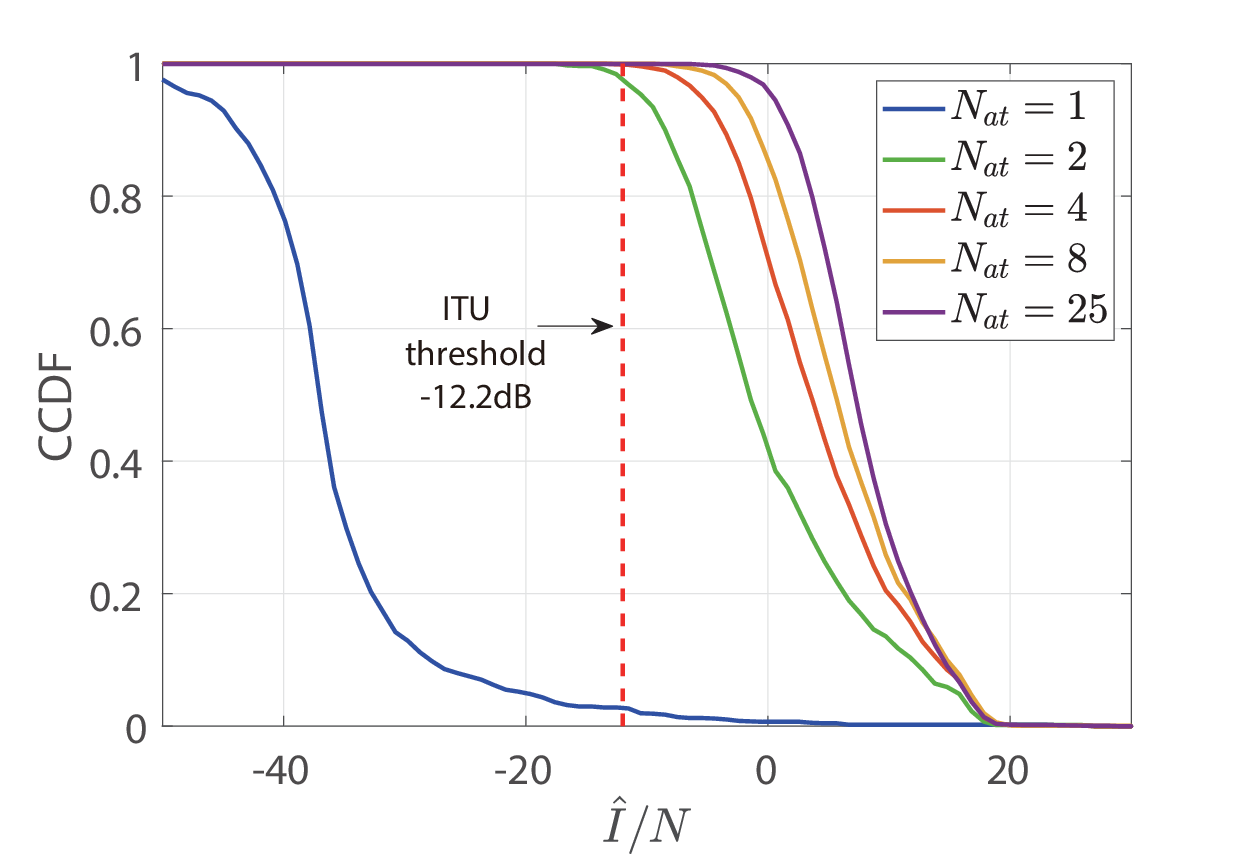}
		\caption{Aggregate I/N distributions under different $N_{\rm at}$ without channelization. The vertical dashed line stands for the {ITU threshold $I_{\rm th}^{\rm R}=-12.2$ dB}.}
		\label{i/nex}
	\end{figure}
	
	\subsection{Problem Formulation}
	
	The above analysis emphasizes the necessity of channelization for IM of mega LEO systems. The goal of channelization-based IM is to properly assign the limited subchannel resources to feeder links with comprehensive consideration of aggregate I/N for all links, so that most harmful interfering link pairs are isolated in frequency domain. Evaluated by the LF occurrence of the system, i.e., the number of feeder links with aggregate I/N greater than {ITU} threshold $I_{\rm th}^{\rm R}$, the optimization problem of channelization-based IM for mega LEO systems at time slot $t$ can be formulated as
	
	\begin{equation}
		\begin{aligned}
			\text{P1:} \quad &\min_{\substack{{\bf c}}} \ 
			f_{\rm LF}({\bf c},t)=\sum_{s  \in \mathcal{W}(t)} U\left( (\hat{I}/N)_{s,t}-I_{\rm th}^{\rm R} \right)\\
			&\text{s.t.} \quad
			c_{s} \in \mathcal{C}, \quad \forall s \in \mathcal{W}(t), \\
		\end{aligned}
	\end{equation}
	where ${\bf c} = [c_{s_1},c_{s_2},...,c_{s_{|\mathcal{W}(t)|}}]$ represents the subchannel assignment scheme and $U(\cdot)$ represents the heaviside function. The minimum possible value of $f_{\rm LF}({\bf c},t)$ is 0, which means perfect IM of the entire system. 
	
	Problem P1 is an NP-hard integer nonlinear programming (INLP) problem \cite{nphard}. Though it can be tackled by optimization solvers, the time complexity grows exponentially with the problem size, imposing a heavy computational burden on satellite operators.
	
	\section{Graph Coloring-Based \\ Problem Transformation}\label{SecG}
	In the considered system, if feeder links, subchannels and strong interference between links are modeled as vertices, colors and edges respectively, then the IM problem based on subchannel allocation can be transformed into a graph coloring problem. Inspired by this, we transform P1 into a $K$-coloring problem, which can be tackled more efficiently.
	
	We first introduce some important definitions related to GC.
	
	\textit{Definition 1: $K$-coloring problem.} Given an undirected graph $G=\{V,E\}$ with vertex set $V$ and edge set $E$, the $K$-coloring problem is to assign a color from set $\{1,2,...,K\}$ to each vertex in $V$, ensuring that adjacent vertices do not share the same color.
	
	\textit{Definition 2: $K$-colorable.} Graph $G$ is $K$-colorable if an aforementioned coloring scheme exists for $K$.
	
	\textit{Definition 3: Chromatic number $\chi(G)$.}
	$\chi(G)$ is the smallest integer $K$ to make $G$ $K$-colorable.
	
	\textit{Definition 4: Degree.} The degree of vertex $x$ is the number of its adjacent vertices, which is denoted as $d(x)$.
	
	Let $G(t)=\{V(t),E(t)\}$ denote the interference graph at time slot $t$, in which $V(t)=\mathcal{W}(t)$ is the vertex set of working satellites at time slot $t$, and $E(t)=\{(u,s)\mid e_{u,s}(t)=1,\forall u,s \in V(t) \}$ represents interference relationships between feeder links of satellites. In order to transform the aggregate interference problem P1 into a GC problem, most existing works use a predefined threshold to determine whether each interfering link is harmful and needs to be modeled on $G(t)$  \cite{pilot1,mimo,routingfassign}. However, this can lead to less accurate assessment of harmful interference in complicated and dynamic interference environments, which are especially intrinsic to mega LEO systems. To this end, we refine it by integrating an adaptive threshold for each link, so that edges are constructed according to the following rule:
	\begin{equation}
		e_{u,s}(t)=
		\begin{cases}
			1, \quad \forall s \in \mathcal{W}(t), \forall u\in \mathcal{S}^I_{s}(c_s,t), I_{u, s,t} \ge I_{\rm th}^{s}, \\
			0, \quad {\rm otherwise},
		\end{cases}
		\label{edge}
	\end{equation}
	where the adaptive interference threshold $I_{\rm th}^{s}$ for satellite $s$ is determined as the maximum value satisfying 
	\begin{equation}
		\sum_{\substack{
				u\in \mathcal{S}^I_{s}(c_s,t) \
				I_{u, s,t} \leq I_{\rm th}^{s} \\
				I_{u, s,t} \geq I_{\rm th}^{\Gamma}
		}} I_{u, s,t} \leq I_{\rm th}^{\rm R} \quad \forall s \in \mathcal{W}(t),
		\label{Iths}
	\end{equation}
	in which $I_{\rm th}^{\Gamma}$ is the weak interference threshold for the whole system. $I_{\rm th}^{\Gamma}$ is the predefined  threshold that excludes numerous weak interfering links, and $I_{\rm th}^{s}$ is designed to adaptively capture almost all harmful interference of $s$ with the fewest possible edges. Fig. \ref{graph} explains the construction of $G(t)$. On this basis, the number of conflicting edges for subchannel assignment scheme (coloring scheme) ${\bf c}$ is calculated by
	\begin{equation}
		f_{\rm con}({\bf c},t)=\sum_{s  \in \mathcal{W}(t)} \sum_{u  \in \mathcal{W}(t)} e_{u,s}(t) \delta(c_{s}-c_{u}),
	\end{equation}
	where $\delta(x)=1$ if $x=0$ and 0 otherwise, then P1 is transformed into 
	\begin{equation}
		\begin{aligned}
			\text{P2:} \quad &\min_{\substack{{\bf c}}}
			\ f_{\rm con}({\bf c},t) \\
			&\text{s.t.} \quad
			c_{s} \in \mathcal{C}, \quad \forall s \in \mathcal{W}(t). \\
		\end{aligned}
	\end{equation}
	P2 is a $C$-coloring problem aiming at minimizing the number of conflicting edges on $G(t)$, thereby mitigating most of the strong interference and reducing the likelihood of LF.
	
	\begin{figure}[t]
		\centering
		\includegraphics[width=0.48\textwidth]{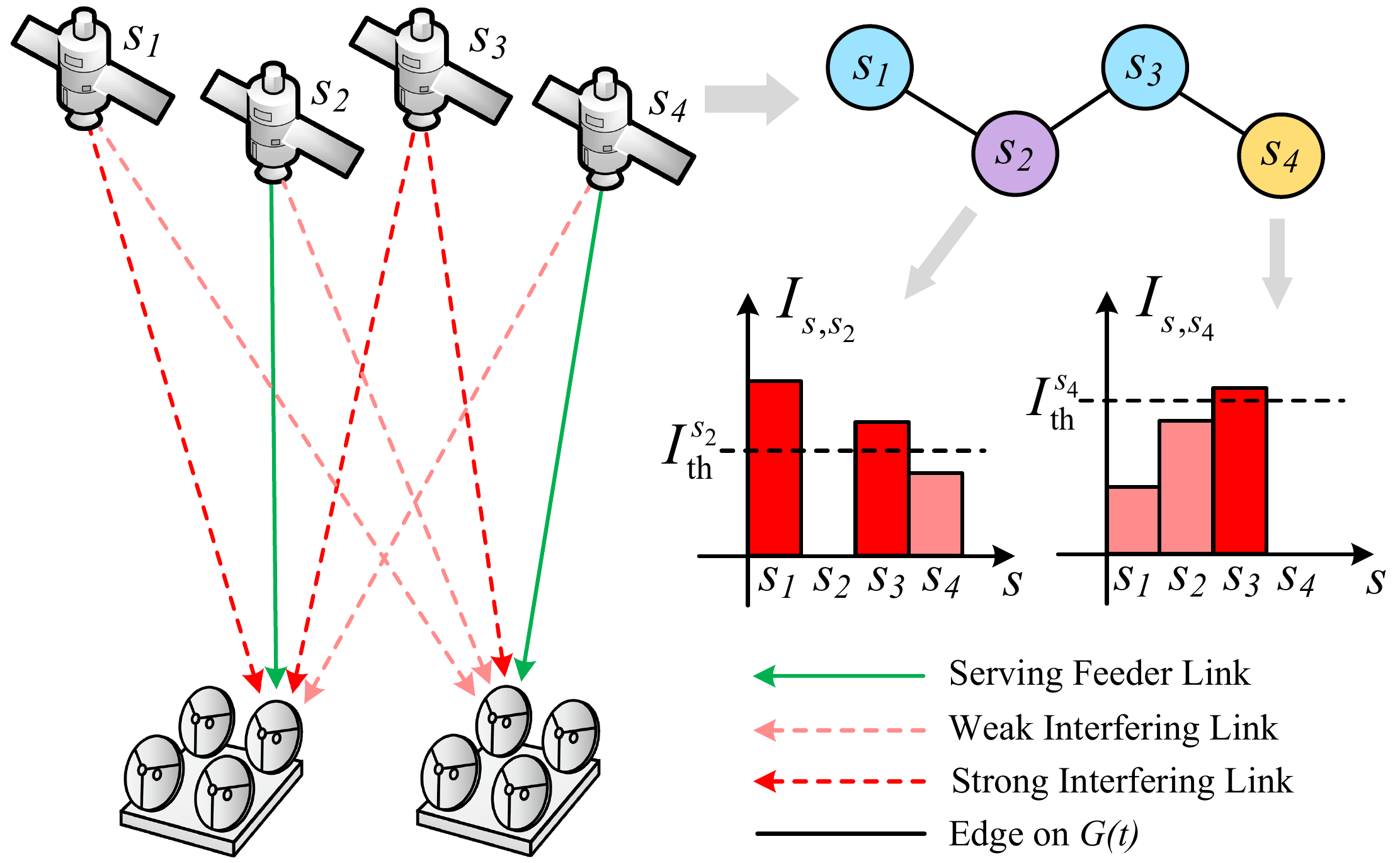}
		\caption{Adapitve edge construction of satellite $s_2$ and $s_4$ on interference graph $G(t)$, in which each satellite has its own adaptive threshold. If an edge exists between two vertices, they should use different colors to avoid conflict, i.e., their corresponding satellites should communicate in different subchannels to avoid harmful interference. }
		\label{graph}
	\end{figure}
	
	\section{Time-Continuous Graph Coloring Algorithms for Dynamic Interference Mitigation}\label{SecA}
	
	In this section, we propose two time-continuous graph coloring algorithms to solve the dynamic IM problem of mega LEO systems. Section \ref{S-GG} and \ref{S-CTS} elaborate the Generalized Global (GG) algorithm and the Clique-Based Tabu Search (CTS) algorithm respectively, with GG prioritizing time efficiency and CTS offering superior IM performance. In Section \ref{S-CFA} we modify the above GC algorithms to realize TCFA without obviously  compromising IM performance. Decomposition of mega constellations aiming at reducing complexity is further discussed in Section \ref{S-CD}.
	
	\subsection{Generalized Global Algorithm}\label{S-GG}
	
	\textit{Global} is a classical GC algorithm renowned for both low complexity and high coloring quality in handling large practical graphs \cite{wp,dygcp}. {
		The procedure of \textit{Global} is outlined as \algref{alg0}. The two key aspects of \textit{Global} are the coloring order and the vertex coloring strategy. \textit{Global} colors vertices in descending order of their degrees (line 2 and line 4) based on the simple observation that vertices with larger degrees should be colored earlier since they have less coloring freedom, and assign each vertex a color that is not used by any of its neighbours (line 5) to ensure that conflicts are avoided locally.}
	
	However, \textit{Global} fails to solve the conflict minimization problem when $G(t)$ is not $C$-colorable in SatCom systems with limited spectrum resources.\footnote{In SatCom systems, excessively dividing subchannels by increasing $C$ will decline spectrum efficiency, thus degrading system capacity. This issue will be further discussed later.} In addition, since $f_{\rm LF}({\bf c},t)$ and $f_{\rm con}({\bf c},t)$ are not strictly equivalent, the optimal solution for P2 could be suboptimal for P1.
	
	Based on the above considerations, we propose the Generalized Global algorithm to effectively reduce LF in our scenario. First, we enhance the coloring strategy of \textit{Global} to address conflict minimization. Second, we diversify the solution exploration of P2 by experimenting with various coloring orders, thereby increasing the chance of finding better solutions for P1. Finally, we add a global coloring adjustment stage to mitigate the inefficiencies of the coloring strategy based only on local information.
	
	GG algorithm involves $N_{\rm GG}$ independent optimizations, selecting the best outcome as the final solution. Each optimization comprises three stages. At stage 1, the degree vector of all  vertices is calculated and slightly perturbed by adding a Gaussian noise vector $\bf n$ with $n[s] \sim \mathcal{N}_{\rm Gau}(0,\sigma^2), \forall s$, i.e.,
	\begin{equation}
		\bf d'=\bf d+\bf n.
	\end{equation}
	Then, sort all vertices in descending order based on $\bf d'$. Since $\bf n$ is different for each optimization, multiple promising solutions for P1 obtained through various coloring orders can be explored. Subsequently, apply \textit{Global} on $G(t)$. 
	
	Stage 2 employs a generalized vertex coloring strategy to reduce conflicts for uncolored vertices, and they are also processed in descending order of $\bf d'$. Each uncolored vertex $s$ is dyed by a least used color 
	\begin{equation}
		c_s= \underset {c \in \mathcal{C}}{\operatorname {arg\,min}   }  \sum_{u\in \mathcal{N} (s)} \delta(c_u-c),
		\label{GG}
	\end{equation}
	instead of a vacant color, so that local conflict minimization is achieved.\footnote{Ties are broken randomly or by choosing the color with the smallest index.} It is worth noting that by addressing conflicts after the non-conflicting \textit{Global} in stage 1, the uncertainty of uncolored neighbors for each vertex is minimized, so that more comprehensive decision can be made for conflict reduction. 
	
	\begin{algorithm}[t]
		\caption{Global Algorithm}
		\begin{algorithmic}[1]
			\STATE \textbf{Input}: $G(t)$. 
			\STATE Calculate degree vector $\bf d$ on $G(t)$ by Definition 4.
			\STATE Initialize ${\bf c}={\bf 0}$ .
			\STATE \textbf{for} $s \in V(t)$ in descending order of $d(s)$ \textbf{do}
			\STATE \quad $c_s=\min\{k \mid k \in \mathbb{Z^+}, k \notin \bigcup\limits_{u \in \mathcal{N}(s)} c_u  \}$.
			\STATE \textbf{end for}
			\STATE \textbf{Output}: $\bf c$.
		\end{algorithmic}
		\label{alg0}
	\end{algorithm}
	
	At stage 3, the colors of all vertices are decided again according to strategy (\ref{GG}) in ascending order of  $\bf d'$ based on the coloring scheme obtained in previous stages, which aims at further reducing conflicts. We design this stage based on the observation that the latter colored vertices in stage 2 tend to concentrate conflicts. By redistributing conflicts more evenly on the whole graph in the reverse order of previous stages, the influence of vertex coloring order is compensated and local inefficiencies are mitigated. 
	
	The complexity of GG is $\mathcal{O}\left( N_{\rm GG} (|\mathcal{W}(t)|+|E(t)|) \right)$ since the generalized coloring strategy (\ref{GG}) does not introduce extra computational overhead compared to \textit{Global}. By observing that coloring vertices by different orders is independent of each other, the $N_{\rm GG}$ optimizations of GG can be parallelized, thereby reducing the parallel complexity of GG to only $\mathcal{O}\left( |\mathcal{W}(t)|+|E(t)| \right)$. 
	
	\subsection{Clique-Based Tabu Search Algorithm}\label{S-CTS}
	While GG is universally applicable to all the conflict minimization GC problems, specific properties in particular scenarios can facilitate the development of a tailored algorithm with superior performance. To this end, we propose the Clique-Based Tabu Search (CTS) algorithm, which leverages the distinctive clique structure of $G(t)$ brought by MAGSs to mitigate all the Intra-GSI fundamentally and further reduces Inter-GSI globally by Tabu Search. 
	
	\textit{Definition 5: Clique.} A clique of graph $G$ is a subset of vertices that every two distinct vertices in it are adjacent \cite{cligc}. 
	
	As analyzed in Section \ref{subsecIstr}, strong Intra-GSI of MAGSs is the predominant cause of interference exacerbation of mega LEO systems. Simulations also show that $(I/N)_{s, u}$ exceeds $I_{\rm th}^{\rm R}$ for almost all satellite pairs $s$ and $u$ selected by the same MAGS. Based on this, each group of satellites selected by the same GS forms a clique or quasi-clique (lacking a few edges to be a clique) of size $N_{\rm at}$. Therefore, denoting the satellite selected by the $n$-th antenna of GS $q$ as $s_q^n$,\footnote{In this subsection we index satellites by GSs and antennas, and they belongs to the same set $\mathcal{W}(t)$ as before.} $V(t)$ can be partitioned by
	\begin{equation}
		V(t)=\mathscr{C}_1(t) \cup \mathscr{C}_2(t) \cup ... \cup \mathscr{C}_{Q}(t),
	\end{equation}
	where each subset given by $\mathscr{C}_i(t)=\{s_i^1,s_i^2,...,s_i^{N_{\rm at}}\}$
	corresponds to the clique or quasi-clique of the $i$-th MAGS. For the sake of generality, $s_q^n$ is treated as a virtual satellite if its corresponding GS antenna is vacant. Vertices of virtual satellites have no adjacent edges and can be colored arbitrarily.
	
	After the clique partition (CP) of $G(t)$, the coloring of each clique can be easily acquired by the following lemma. 
	
	\begin{lemma}
		For a clique $\mathscr{C}$ with $C$ vertices, it holds that $\chi(\mathscr{C})=C$, i.e., all conflicts can be avoided if and only if each vertex is assigned a different color from the set $\{1,2,...,C\}$.\footnote{Typically, removing a few edges from a clique does not alter its chromatic number, hence this lemma is still valid for most quasi-cliques.}
		\label{lemma1}
	\end{lemma}
	
	\lmref{lemma1} and the CP result of $V(t)$ suggest that $C \geq N_{\rm at}$ is a necessary condition for system-level IM. In this paper, we set $C=N_{\rm at}$ to achieve excellent balance between accessibility and system capacity, which is validated in Section \ref{SecS}. 
	
	Based on Lemma 1, the Intra-GSI of the entire system can be completely eliminated by the coloring solution structure
	\begin{equation}
		(c_{s_i^1},c_{s_i^2},...,c_{s_i^{N_{\rm at}}})= \pi(1,2,...,N_{\rm at}), \quad \forall i \in \mathcal{Q},
		\label{struct}
	\end{equation}
	where $\pi(1,2,...,N_{\rm at})$ is an arbitrary permutation of $(1,2,...,N_{\rm at})$. Subsequently, considering the effectiveness of Tabu Search on GC \cite{tabu,head}, we further mitigate Inter-GSI by applying Tabu Search to the initial solutions generated by structure (\ref{struct}). 
	
	\begin{figure}[t]
		\centering
		\includegraphics[width=0.45\textwidth]{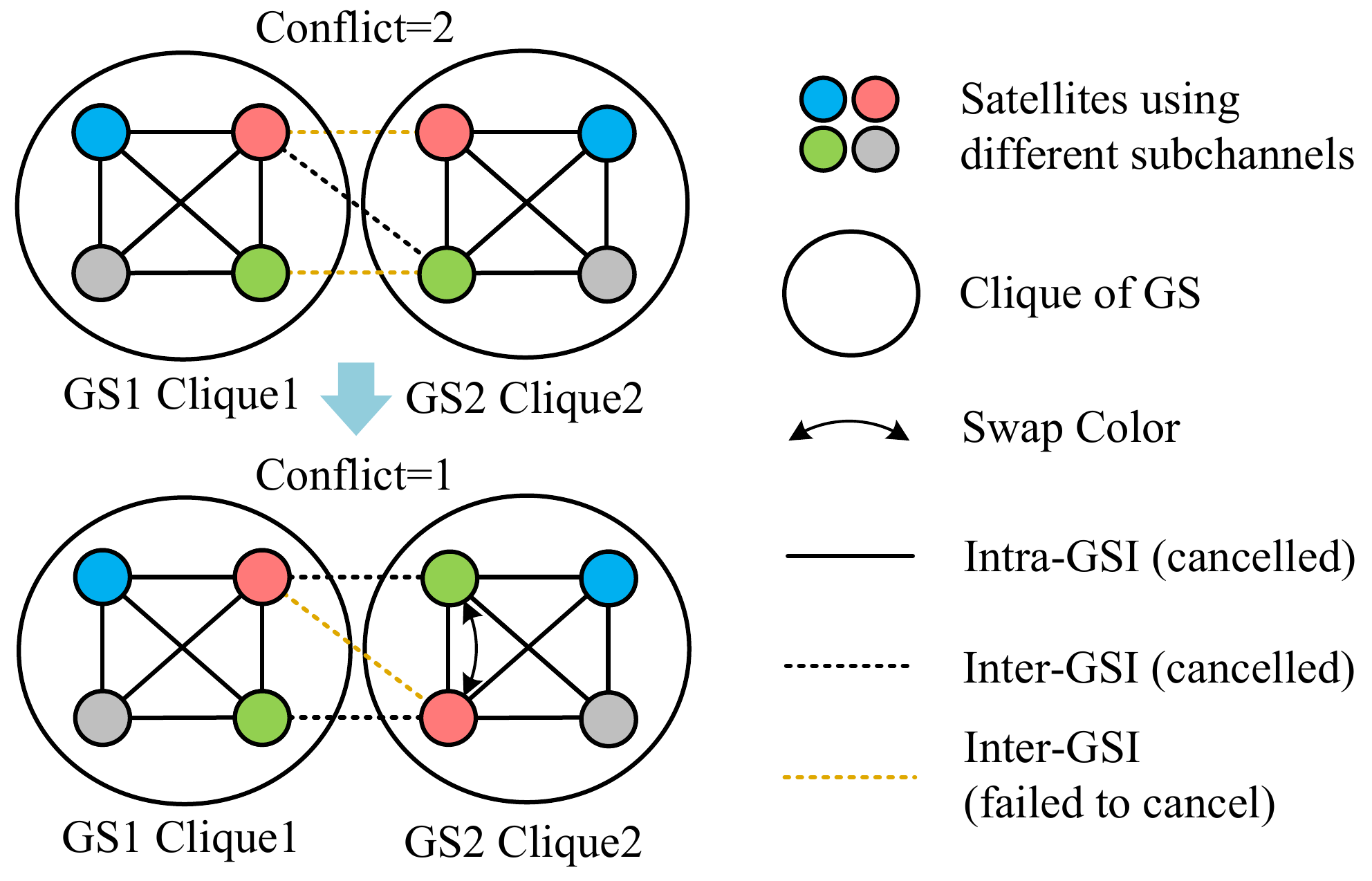}
		\caption{Illustration of Lemma 1, solution structure (\ref{struct}) and neighboring solution generation strategy (\ref{change}) of CTS algorithm. The red and green vertices of clique 2 are randomly selected and their colors are swapped to generate a neighboring solution, which reduces the number of Intra-GSI conflicts from 2 to 1.}
		\label{nei}
	\end{figure}
	
	Tabu Search is a metaheuristic local search algorithm that iteratively navigates through the solution space and avoids cyclic search by utilizing a tabu list \cite{tabu}. During each iteration of Tabu Search in our scenario, we first generate $N_{\rm n}$ new coloring schemes ${\bf C'} =({\bf c}^{(1)},{\bf c}^{(2)},...,{\bf c}^{(N_{\rm n})})$ (neighboring solutions) by swapping the colors of two vertices belonging to the same clique on the current best solution ${\bf c}^*$. Formally, if $q$ is a random GS and $x,y$ are two different random antennas, then a neighboring solution ${\bf c}^{(i)}$ can be written as
	
	\begin{equation}
		\begin{aligned}
			{ c}^{(i)}_s=
			\begin{cases}
				{ c}^*_{s_q^y}, & s= {s_q^x},\\
				{ c}^*_{s_q^x}, & s= {s_q^y},\\
				{ c}^*_{s}, &  {\rm otherwise},
			\end{cases}
			\forall s \in \mathcal{W}(t).
		\end{aligned}
		\label{change}
	\end{equation}
	This neighboring solution generation strategy ensures that sturcture \eqref{struct} is preserved throughout the searching process, which is illustrated in Fig. \ref{nei} using a simplified example. It should be noted that we only select vertices with conflicts as ${s_q^x}$, so that ineffective searches can be reduced. To avoid revisiting previously explored solutions, a tabu list $L=\{{\bf c}^*_{-1},{\bf c}^*_{-2},...,{\bf c}^*_{-l_{\rm tabu}}\}$ records the solutions of the last $l_{\rm tabu}$ steps at each iteration, where ${\bf c}^*_{-i}$ represents the best solution acquired $i$ iterations earlier. Then at the end of each iteration, ${\bf c}^*$ is updated as
	\begin{equation}
		{\bf c}^* = \underset {{\bf c} \in {\bf C'}, {\bf c} \notin L}{\operatorname {arg\,min} } f_{\rm con}({\bf c},t),
	\end{equation}
	and $L$ is updated as $\{{\bf c}^*,{\bf c}^*_{-1},...,{\bf c}^*_{-l_{\rm tabu}+1}\}$ accordingly. 
	
	Based on the above descriptions, CTS algorithm with two stages can be concluded: at stage 1, generate and test $N^{t}_{\rm in}$ random initial solutions according to sturcture \eqref{struct}, then select the $N^{t}_{\rm ca}$ best candidates among them; at stage 2, optimize these candidate solutions using Tabu Search and output the best solution with the lowest $f_{\rm LF}({\bf c},t)$.
	
	Similar to GG, both stages of CTS can be parallelized due to the independence of generating and testing each initial solution, as well as performing Tabu Search on different candidate solutions. Therefore, with at least $N^{t}_{\rm in}$ ($N^{t}_{\rm in}>N^{t}_{\rm ca}$) parallel executors, the parallel complexity of CTS can be reduced from $\mathcal{O}\left( N^{t}_{\rm in}|\mathcal{W}(t)|+N^{t}_{\rm ca}N_{\rm it}t_{\rm tabu}\right)$ to $\mathcal{O}\left( |\mathcal{W}(t)|+N_{\rm it}t_{\rm tabu}\right)$, where $t_{\rm tabu}$ represents  the complexity per iteration of Tabu Search. 
	
	\subsection{Time Continuous Frequency Allocation for Dynamic Constellation}\label{S-CFA}
	In dynamic SatCom systems, it is of great significance to ensure TCFA for feeder links, so that excessive additional signaling overhead of switching frequency can be avoided. Specifically, when a satellite is selected by the same GS for a consecutive period of time, the subchannel allocated to it should not change frequently, i.e., if $q_s^t=q_s^{t-1}$ holds, $c_s^t=c_s^{t-1}$ should also be maintained as much as possible, where superscript $(\cdot)^t$ conveys the meaning of ``at time slot $t$" in this subsection. For the convenience of expression, we define $m_{\rm c}(s,t)=1$ if $q_s^t=q_s^{t-1}$ and 0 otherwise.
	
	\begin{figure}[t]
		\centering
		\includegraphics[width=0.45\textwidth]{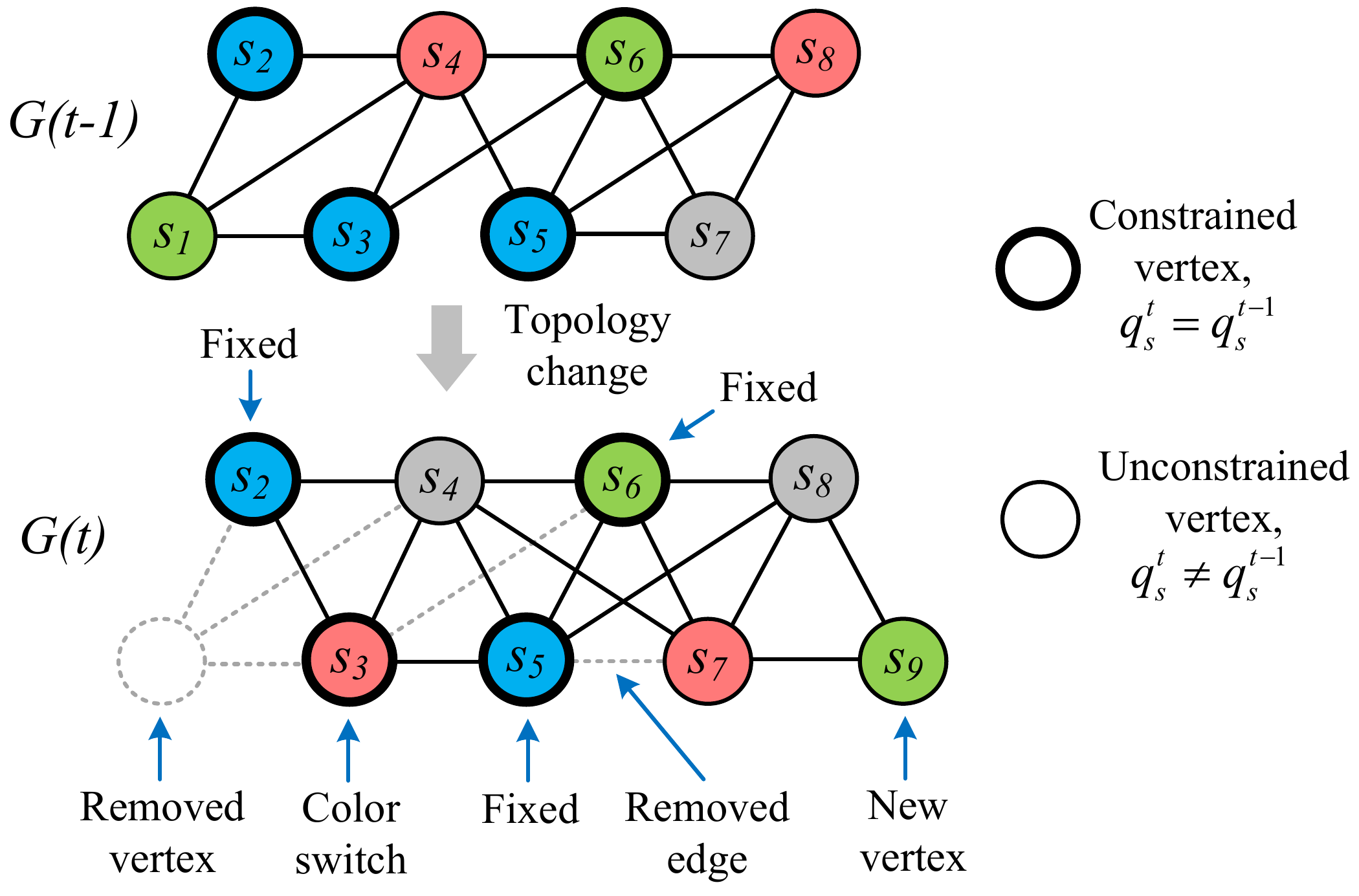}
		\caption{An Example of GC with low FSR. Due to the dynamic nature of LEO constellations, the topology of the interference graph is time-variant. When edges $(s_2,s_3)$ and $(s_3,s_5)$ emerge at time $t$, the color of constrained vertex $s_3$ needs to be switched with a high probability since it would otherwise cause 2 conflicts. The colors of other constrained vertices remain fixed, and other unconstrained vertices can adjust their colors freely to avoid conflicts. }
		\label{cfa}
	\end{figure}
	
	Reflecting TCFA on GC, the colors of a certain set of vertices $\mathcal{S}_{\rm c}(t)=\{s \mid q_s^t=q_s^{t-1} , s\in \mathcal{W}(t) \}$ on the current graph $G(t)$ are constrained by their previous colors in ${\bf c}^{t-1}$ on $G(t-1)$, which poses additional challenges to GC algorithm design. To this end, we further modify GG and CTS to achieve a low frequency switching rate (FSR) while maintaining high IM performance. The key idea of modification involves inheriting colors of all constrained  vertices (in $\mathcal{S}_{\rm c}(t)$) from ${\bf c}^{t-1}$ before GC, i.e., 
	\begin{equation}
		c_s^t=c_s^{t-1}, \forall s \in \mathcal{S}_{\rm c}(t),
		\label{inher}
	\end{equation}
	and switching their colors with a probability proportional to the associated switching benefits (conflict reduction) during GC, as explained by the example in Fig. \ref{cfa}. Specifically, let $\bf c$ and $\bf c'$ represent the coloring schemes before and after an operation that switches the colors of $k_{\rm c}$ constrained vertices and an arbitrary number of unconstrained vertices, then the probability of accepting this operation is given by
	\begin{equation}
		\begin{aligned}
			p({\bf c},{\bf c'}) = 
			\min\{ \Delta f_{\rm con} p_{\rm s}/(k_{\rm c}+\epsilon), 1 \}, 
		\end{aligned}
		\label{swicol}
	\end{equation}
	where $\Delta f_{\rm con} = \max\{f_{\rm con}({\bf c},t)-f_{\rm con}({\bf c'},t),0\} $ is the number of reduced conflicts, $\epsilon$ is a very small positive number, and $p_{\rm s}$ is the proportionality constant introduced to realize a controllable long-time average FSR for the system. On this basis, we refine GG and CTS as follows: 
	
	\textit{1) TCFA-GG:} First, inherit colors of constrained vertices according to \eqref{inher} and leave unconstrained vertices uncolored, rendering $G(t)$ partial-colored. Next, sort unconstrained vertices in descending order of their degrees and execute stage 1 and 2 of GG to color them. For stage 3, switch the colors of constrained vertices first according to strategy \eqref{GG} with the probability given by \eqref{swicol}, in which $k_{\rm c}=1$ since each operation only recolors one constrained vertex. Then, recolor other unconstrained vertices by strategy \eqref{GG}. TCFA-GG only switches the colors of a few constrained vertices that would otherwise cause severe interference, and the impact of fixing them is offset by recoloring other unconstrained vertices flexibly. The process of TCFA-GG is presented in \algref{alg1}.\footnote{TCFA-GG and TCFA-CTS degrades to GG and CTS when $\mathcal{S}_{\rm c}(t)=\emptyset$.} 
	
	\begin{algorithm}[t]
		\caption{TCFA-GG Algorithm}
		\begin{algorithmic}[1]
			\STATE \textbf{Input}: $G(t), N_{\rm GG}, C,  \mathcal{S}_{\rm c}(t)$, ${\bf c}^{t-1}$. 
			\STATE Define ${\bf c}^*$ to record the best solution.
			\STATE Calculate degree vector $\bf d$ on $G(t)$ by Definition 4.
			\STATE \textbf{for} $i=1,...,N_{\rm GG}$ \textbf{do} (in parallel)
			\STATE \quad Initialize ${\bf c}^{t}={\bf 0}$.
			\STATE \quad $c_s^t=c_s^{t-1}, \forall s \in \mathcal{S}_{\rm c}(t).$
			\STATE \quad \textbf{Stage 1:}
			\STATE \quad $\bf d'=\bf d+\bf n$.
			
			\STATE \quad \textbf{for} $s \in V(t) \setminus \mathcal{S}_{\rm c}(t)$ in descending order of $d'(s)$ \textbf{do}
			\STATE \quad \quad \textbf{if} $	\underset {c \in \mathcal{C}}{\operatorname {\min}   }  \sum_{u\in \mathcal{N} (s)} \delta(c_u-c)=0$ \textbf{do}
			\STATE \quad \quad \quad $	c_s^t= \underset {c \in \mathcal{C}}{\operatorname {arg\,min}   }  \sum_{u\in \mathcal{N} (s)} \delta(c_u-c)$.
			\STATE \quad \quad \textbf{end if} 
			\STATE \quad\textbf{end for}
			
			\STATE \quad \textbf{Stage 2:}
			\STATE \quad \textbf{for} $s \in V(t) \setminus \mathcal{S}_{\rm c}(t)$ in descending order of $d'(s)$ \textbf{do}
			
			\STATE \quad \quad \textbf{if} $	c_s^t= 0$ \textbf{do}
			\STATE \quad \quad \quad $	c_s^t= \underset {c \in \mathcal{C}}{\operatorname {arg\,min}   }  \sum_{u\in \mathcal{N} (s)} \delta(c_u-c)$.
			\STATE \quad \quad \textbf{end if}
			\STATE \quad \textbf{end for}
			
			\STATE \quad \textbf{Stage 3:}
			\STATE \quad \textbf{for} $s \in \mathcal{S}_{\rm c}(t)$ in ascending order of $d'(s)$ \textbf{do}
			\STATE \quad \quad Do $c_s^t= \underset {c \in \mathcal{C}}{\operatorname {arg\,min}   }  \sum_{u\in \mathcal{N} (s)} \delta(c_u-c)$ with the \\ \quad \quad probability given by \eqref{swicol}, in which $k_c=1$.
			\STATE \quad \textbf{end for}
			
			\STATE \quad \textbf{for} $s \in V(t) \setminus \mathcal{S}_{\rm c}(t)$ in ascending order of $d'(s)$ \textbf{do}
			\STATE \quad \quad $	c_s^t= \underset {c \in \mathcal{C}}{\operatorname {arg\,min}   }  \sum_{u\in \mathcal{N} (s)} \delta(c_u-c)$.
			\STATE \quad \textbf{end for}
			\STATE \quad \textbf{if} $f_{\rm LF}({\bf c}^t,t) < f_{\rm LF}({\bf c}^*,t)$ \textbf{do}
			\STATE \quad \quad ${\bf c}^* = {\bf c}^t$.
			\STATE \quad \textbf{end if}
			\STATE \textbf{end for}
			\STATE ${\bf c}^t={\bf c}^*$.
			\STATE \textbf{Output}: ${\bf c}^t$.
		\end{algorithmic}
		\label{alg1}
	\end{algorithm}
	
	\textit{2) TCFA-CTS:} Similar to TCFA-GG, the first step of TCFA-CTS is also color inheritance of constrained vertices. Then, for generating initial solutions at stage 1 of CTS, the colors of constrained vertices are fixed, while those of unconstrained vertices in the same clique can be permuted freely. In stage 2, conflicts of constrained and unconstrained vertices are reduced together. Specially, a neighboring solution generated by swapping the colors of $s_q^x$ and $s_q^y$ is saved into ${\bf C'}$ with the probability determined by \eqref{swicol}, in which $k_c=m_c(s_q^x,t)+m_c(s_q^y,t)$.
	New neighboring solutions will continue to be generated until their number reaches $N_{\rm n}$ for each iteration of Tabu Search. The process of TCFA-CTS is summarized in \algref{alg2}.
	
	\begin{algorithm}[t]
		\caption{TCFA-CTS Algorithm}
		\begin{algorithmic}[1]
			\STATE \textbf{Input}: $G(t)$, $N^{t}_{\rm in}$, $N^{t}_{\rm ca}$, $N_{\rm it}$, $N_{\rm n}$, $C$, $\mathcal{S}_{\rm c}(t)$, ${\bf c}^{t-1}$.
			
			\STATE \textbf{Stage 1:}
			\STATE \textbf{for} $i=1,...,N^{t}_{\rm in}$ \textbf{do} (in parallel)
			\STATE \quad Initialize ${\bf c}={\bf 0}$.
			\STATE \quad $c_s=c_s^{t-1}, \forall s \in \mathcal{S}_{\rm c}(t).$
			\STATE \quad \textbf{for} $q=1,...,Q$ \textbf{do}
			\STATE \quad \quad $\mathcal{P}=\mathcal{C} \setminus \{c \mid c \in (c_{s_q^1},c_{s_q^2},...,c_{s_q^C})\}$.
			\STATE \quad \quad Generate $\bf p$ by randomly permuting elements in $\mathcal{P}$.
			\STATE \quad \quad Insert elements of vector $\bf p$ sequentially into positions \\ \quad \quad of zero elements in $(c_{s_q^1},c_{s_q^2},...,c_{s_q^C})$.
			\STATE \quad \textbf{end for}
			\STATE \quad Let the $i$-th initial solution ${\bf c}_{i}={\bf c}$.
			\STATE \quad Calculate $f_{\rm con}({\bf c}_{i},t)$.
			\STATE \textbf{end for}
			\STATE Select $N^{t}_{\rm ca}$ initial solutions with the smallest $f_{\rm con}({\bf c}_i,t)$ and  record them as $(\hat{{\bf c}}_{1},...,\hat{{\bf c}}_{N^{t}_{\rm ca}})$. 
			
			\STATE \textbf{Stage 2:}
			\STATE \textbf{for} $i'=1,...,N^{t}_{\rm ca}$ \textbf{do} (in parallel)
			\STATE \quad ${\bf c}^*=\hat{{\bf c}}_{i'}$. 
			\STATE \quad Initialize $L$ as empty tabu list and  counter $j=0$.
			\STATE \quad \textbf{for} ${\rm it}=1,...,N_{\rm it}$ \textbf{do}
			\STATE \quad \quad \textbf{while} $j<N_{\rm n}$ \textbf{do}
			\STATE \quad \quad \quad Randomly select a vertex $s_q^x$ with conflicts and $s_q^y$. 
			
			\STATE \quad \quad \quad Swap the colors of them by \eqref{change} to generate ${\bf c}^{(j)}$.
			\STATE \quad \quad \quad 	$k_c=m_c(s_q^x,t)+m_c(s_q^y,t)$.
			\STATE \quad \quad \quad With the probability according to \eqref{swicol}, save ${\bf c}^{(j)}$  \\ \quad \quad \quad into ${\bf C'}$ and do $j=j+1$. 
			
			\STATE \quad \quad \textbf{end while}
			
			\STATE \quad \quad ${\bf c}^* = \underset {{\bf c} \in {\bf C'}, {\bf c} \notin L}{\operatorname {arg\,min} } f_{\rm con}({\bf c},t)$.
			\STATE \quad \quad Update $L$.
			\STATE \quad \textbf{end for}
			\STATE \quad \textbf{if} $f_{\rm LF}({\bf c}^*,t) < f_{\rm LF}({\bf c}^t,t)$ \textbf{do}
			\STATE \quad \quad ${\bf c}^t = {\bf c}^*$.
			\STATE \quad \textbf{end if}
			\STATE \textbf{end for}
			\STATE \textbf{Output}: ${\bf c}^t$.
		\end{algorithmic}
		\label{alg2}
	\end{algorithm}
	
	\subsection{Complexity Reduction by Mega Constellation Decomposition}\label{S-CD}
	In order to further alleviate the complexity burden of SatCom operators, we propose two mega constellation decomposition methods in this subsection, which adapt to different system scales. By dividing $G(t)$ into small subgraphs with low or no correlation, GC algorithms can be executed on them in parallel, thereby reducing the overall time complexity to that on the largest subgraph. Two decomposition methods are detailed as follows.
	
	\textit{1)} \textit{Connected Component Decomposition (CCD):} Since interference between distant GSs are weak \cite{gs1}, $G(t)$ of the entire system is unconnected so that it can be decomposed into several connected components (CC) by a simple Depth-First Search (DFS) algorithm as
	\begin{equation}
		G(t)=\{G^{\rm c}_1(t),G^{\rm c}_2(t),...,G^{\rm c}_A(t)\},
	\end{equation}
	where $A$ is the number of CCs. 
	
	CCD is lossless for GC since no edges are deleted, and GC on different CCs are independent. However, the effect of CCD diminishes as $N_{\rm at}$ grows since stronger interference increases the scale of the largest CC. This calls for  decomposition inside huge CCs to maintain the effect of complexity reduction.
	
	\textit{2)} \textit{Gateway Station Clustering Decomposition (GSCD):} 
	Inspired by the CP of $V(t)$ in CTS algorithm, we propose the gateway station clustering decompositon method, which is a generalization of CCD.
	
	GSCD categorize GSs into clusters by the K-means algorithm \cite{kmeans}, using the Earth-centered, Earth-fixed (ECEF) coordinates of GSs as features and the within-cluster sum of squares (WCSS) minimization principle as the criterion. 
	Subsequently, all satellites selected by a cluster of GSs form a subgraph, so that $G(t)$ is partitioned as 
	\begin{equation}
		\begin{aligned}
			&V(t)=V^{\rm p}_1(t) \cup V^{\rm p}_2(t) \cup... \cup V^{\rm p}_{A_1}(t),\\
			&E(t)=E^{\rm p}_1(t) \cup E^{\rm p}_2(t) \cup...\cup E^{\rm p}_{A_1}(t)\cup D^{\rm p}(t),
		\end{aligned}
		\label{gp}
	\end{equation}
	where $A_1$ is the number of GS clusters, and
	\begin{equation}
		\begin{aligned}
			&V^{\rm p}_i(t)=\cup_{q\in {\mathcal{Q}^{\rm c}_i}}\mathscr{C}_q(t),\\
			&E^{\rm p}_i(t)=\{(u,s)|u\in V^{\rm p}_i(t),s\in V^{\rm p}_i(t)\},\\
			&D^{\rm p}(t)=E(t) \setminus \left(E^{\rm p}_1(t) \cup E^{\rm p}_2(t) \cup...\cup E^{\rm p}_{A_1}(t)\right),
		\end{aligned}
	\end{equation}
	where $\mathcal{Q}^{\rm c}_i$ is the GS index set representing the $i$-th cluster.
	After GSCD, GC algorithm is parallelized on each subgraph $G^{\rm p}_i(t)=\{V^{\rm p}_i(t),E^{\rm p}_i(t)\}$. 
	
	However, the process of GSCD leads to information loss due to edge removal, hence end vertices of deleted edges belonging to $D^{\rm p}(t)$ need to be recolored for compensation. For GG, these vertices are recolored according to strategy (\ref{GG}), and the recoloring order can be arbitrary since the information of degree is damaged after GSCD. For CTS, recover $G(t)$ and execute the Tabu Search process on the parallel coloring result for several iterations until convergence. 
	
	\section{System Capacity Maximization}\label{SecR}
	\vspace{-1mm}	
	
	In the {previous sections}, we focus on problem P1 and P2 based on the ITU regulation of aggregate I/N. In this section, we further investigate the IM goal of maximizing system capacity, which is another important metric of SatCom systems \cite{worry,densesky,newspace}. {First we explain that} problem P1 and system capacity maximization share a highly aligned objective. {Then we further enhance} spectrum efficiency by vacant subchannel utilization. 
	
	\vspace{-1mm}
	\subsection{Relationship between Graph Coloring-Based Interference mitigation and System Capacity Maximization}\label{S-scm}
	Apart from the I/N metric studied in previous sections, carrier to noise ratio (C/N) and carrier to interference plus noise ratio (SINR) are also important metrics to evaluate the quality of a feeder link, which are defined as
	\begin{equation}
		\label{CtoN}
		(C/N)_{s}=\frac{P_{s}^{\rm tr}}{\kappa T_{\rm n}B}, \quad
		{\rm SINR}_{s}=\frac{P_{s}^{\rm tr}}{\hat{I}_{s}+\kappa T_{\rm n}B},
	\end{equation}
	for satellite $s$, respectively. As a result, the link capacity $R_{s}$ is calculated as
	\begin{equation}
		\begin{aligned}
			R_{s}\!=\!B\log_2(1+{\rm SINR}_{s}),
		\end{aligned}
	\end{equation}
	and the theoretical upper bound of $R_{s}$ is the capacity under no interference, i.e., 
	\begin{equation}
		\begin{aligned}
			\overline{R}_{s}\!=\!B\log_2\left(1+(C/N)_{s}\right).
		\end{aligned}
	\end{equation}
	Similar to \cite{risk,newspace,s1432,densesky}, link capacity degradation $\Delta R_{s}$ is defined as
	\begin{equation}
		\label{deltar}
		\begin{aligned}
			\Delta R_{s}=1-\frac{R_{s}}{\overline{R}_{s}}=1-\frac{\log_2(1+{\rm SINR}_{s})}{\log_2(1+(C/N)_{s})},
		\end{aligned}
	\end{equation}
	and system capacity degradation $\Delta \hat{R}$ is defined  as
	\begin{equation}
		\Delta \hat{R}= 1- \frac{\sum_{s  \in \mathcal{W}(t)} R_{s}}{\sum_{s  \in \mathcal{W}(t)} \overline{R}_{s}} .
	\end{equation}
	
	\begin{ppn}
		When the feeder link of satellite $s$ meets the ITU regulation and ${\rm SINR}_{s}\!>\!1$, the following inequality holds:
		\begin{equation}
			\label{Rleq}
			\Delta R_{s} \leq \log_{(C/N)_{s}}(1+I_{\rm th}^{\rm R}).
		\end{equation}
	\end{ppn}
	\begin{proof}
		Rewrite the regulation of ITU as $\hat{I}_{s}+N \leq (1+I_{\rm th}^{\rm R})N$ for satellite $s$ and substitute it into (\ref{CtoN}), we have
		\begin{equation}
			\label{21}
			{\rm SINR}_{s} \geq \frac{(C/N)_{s}}{1+I_{\rm th}^{\rm R}}.
		\end{equation}
		When ${\rm SINR}_{s}>1$ is satisfied\footnote{This condition can be easily satisfied in modern SatCom systems after IM.}, the following equation holds: 
		\begin{align}
			\delta_R = \frac{\log_2(1+{\rm SINR}_{s})}{\log_2(1+(C/N)_{s})} - \frac{\log_2{\rm SINR}_{s}}{\log_2\left(C/N\right)_{s}}>0,
			\label{dr}
		\end{align}
		which is proved in Appendix \ref{Apenproof}. As a result, we can obtain 
		\begin{equation}
			\begin{aligned}
				\Delta R_{s} 
				&= 1-\frac{\log_2{\rm SINR}_{s}}{\log_2(C/N)_{s}}-\delta_R \\
				&\leq 1- \frac{\log_2(C/N)_{s}-\log_2(1+I_{\rm th}^{\rm R})}{\log_2(C/N)_{s}} \\
				&= \log_{(C/N)_{s}}(1+I_{\rm th}^{\rm R}).
			\end{aligned}
		\end{equation}
		This completes the proof. 	
	\end{proof}
	
	It's noted that $\log_{(C/N)_{s}}(1+I_{\rm th}^{\rm R})$ is a very small value, which indicates that the capacity $R_s$ almost reaches $\overline{R}_{s}$. Therefore, after solving P1 by GC, $\Delta R_s$ of most feeder links are negligible, thereby yielding a small $\Delta \hat{R}$ for the whole system. In other words, GC based IM and system capacity maximization are ultimately aligned in their outcomes.
	
	\vspace{-2mm}
	\subsection{Vacant Subchannel Utilization }\label{S-uvs}
	\vspace{-1mm}
	All the above discussions are based on the assumption that each feeder link occupies exactly one subchannel. However, with the increase of $N_{\rm at}$, some MAGSs will select satellites with larger elevation angle or even leave some antennas vacant. As a result, the number of subchannels needed for IM of these MAGSs can be smaller than $N_{\rm at}$, leading to a waste of frequency resources. Therefore, we propose a vacant subchannel utilization (VSU) method, which reallocates vacant subchannels in the system to some feeder links to double their transmission bandwidth, thereby further improving system capacity.
	
	Based on a subchannel allocation scheme ${\bf c}$ obtained by GC algorithms in the preceding sections, we introduce VSU scheme ${\bf c}^{\rm r}= [c^{\rm r}_{s_1},c^{\rm r}_{s_2},...,c^{\rm r}_{s_{|\mathcal{W}(t)|}}]$, where each element  $c^{\rm r}_{s_i} \in \{0,1,2,...,C\}$ is the vacant subchannel reallocated to satellite $s_i$. In particular, $c^{\rm r}_{s_i}=0$ stands for no subchannel reuse. The goal of VSU is to assign as many vacant subchannels to feeder links as possible without causing more interference, which can be formulated as 
	\begin{equation}
		\begin{aligned}
			\text{P3:} \quad  \max_{{\bf c}^{\rm r}} \ 
			&f_3(t)= \sum_{s  \in \mathcal{W}(t)} U(c^{\rm r}_{s}-\epsilon)\\
			\text{s.t.} \quad
			& c^{\rm r}_{s} \in \{0,1,...,C\}, \quad \forall s \in \mathcal{W}(t), \\
			& c^{\rm r}_{s} \neq c_u, \quad \forall u \in \mathcal{N}(s)\cup\{s\}, \forall s \in \mathcal{W}(t), \\
			& c^{\rm r}_{s} \neq c^{\rm r}_{u}, \quad \forall u \in \mathcal{N}(s), \forall s \in \mathcal{W}(t). \\
		\end{aligned}
	\end{equation}
	
	P3 is a list coloring problem \cite{list} in which each vertex has its own feasible color set and the goal is to color all the vertices without conflicts under feasible color constraints. Unlike the conflict minimization GC problem P2, if coloring vertex $s$ results in conflicts here, $s$ is left uncolored (dyed by color 0) since the regulation of ITU has higher priority. It is observed that $c^{\rm r}_{s}$ of each vertex $s$ is constrained by the colors of itself and all its neighbors, so its feasible color set is  
	\begin{equation}
		\mathscr{L}_s = \mathcal{C} \setminus \biggl( \bigcup\limits_{u \in \mathcal{N}(s) \cup \{s\}} c_u \biggr).
	\end{equation} 
	Similar to TCFA-GG, P3 can be efficiently solved by implementing \textit{Global} on a partial-colored graph  generated as follows: First, form a reuse graph $G'(t)$  by duplicating $G(t)$, and rename its vertices as $s^{\rm r}_1, s^{\rm r}_2,...,s^{\rm r}_{s_{|\mathcal{W}(t)|}}$. Then, inherit the colors of all vertices on $G(t)$ and fix them. Finally, add edges $(s,s^{\rm r}), \forall s \in \mathcal{W}(t)$, and $(u,s^{\rm r})$ if $e_{u,s}(t)=1,\forall u,s \in V(t)$.

	\section{Simulation Results and Performance Analysis}\label{SecS}
	
	\begin{figure*}[t]
		\centering
		\subfloat[]{
			\includegraphics[scale=0.25]{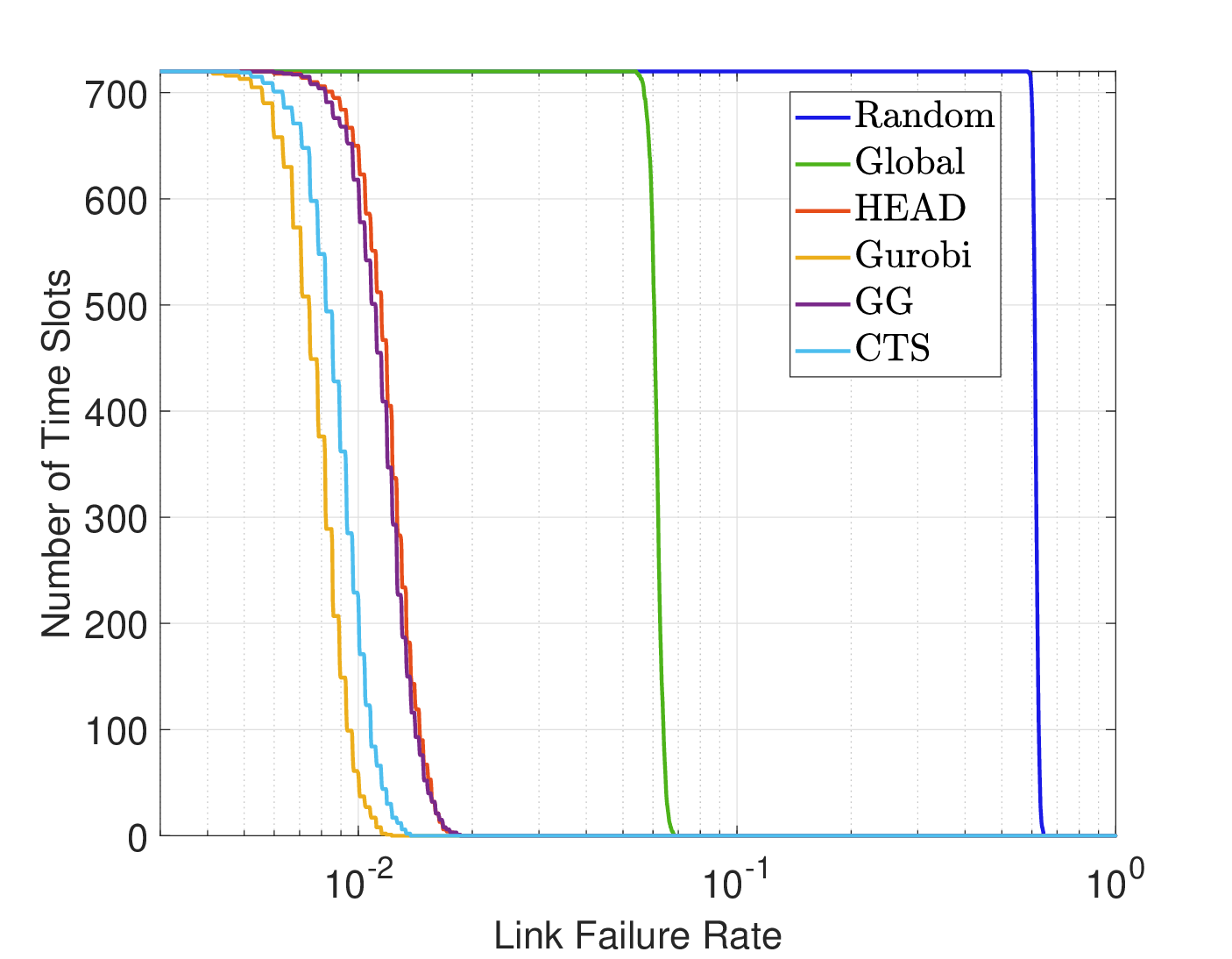}}
		\subfloat[]{
			\includegraphics[scale=0.25]{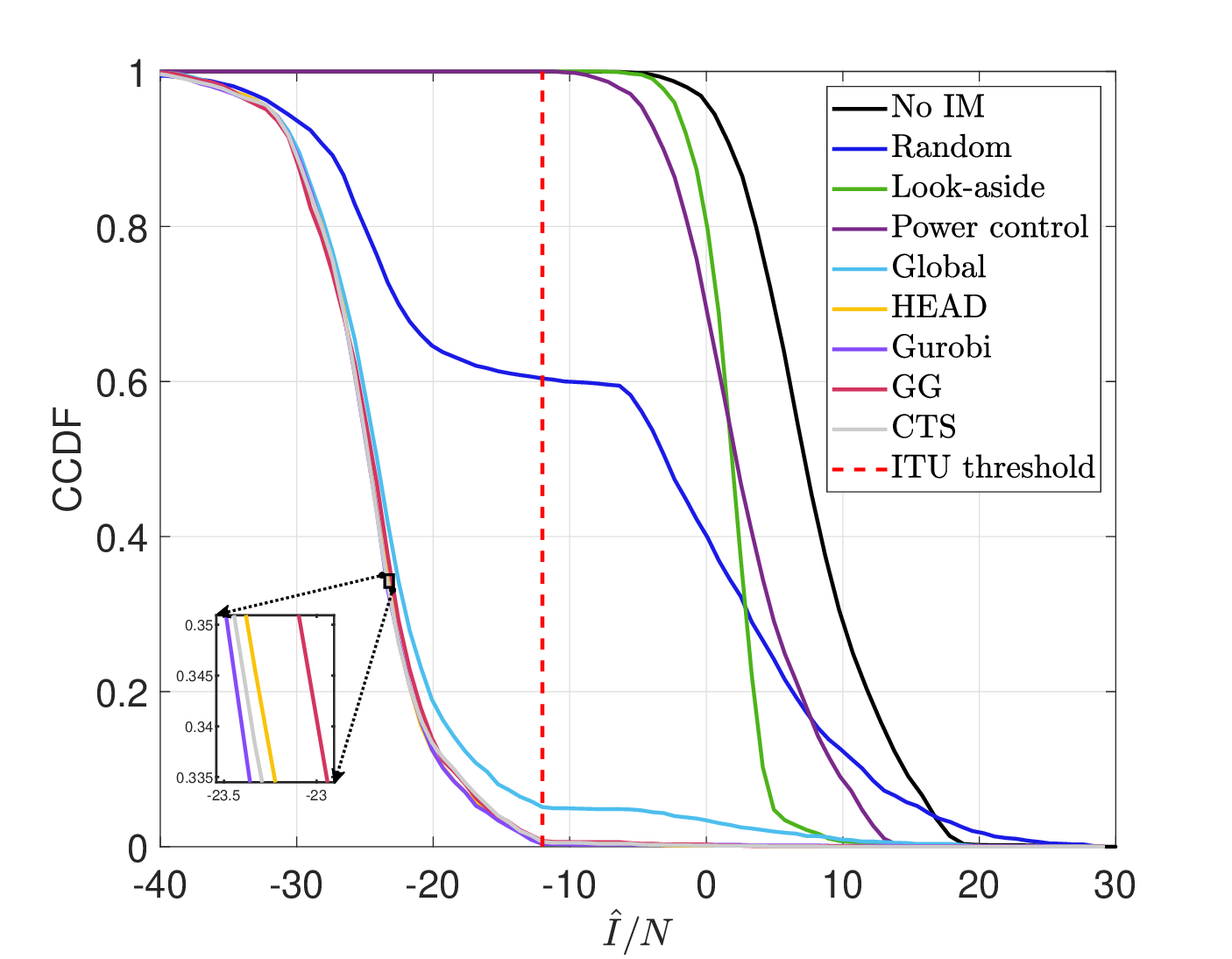}}
		\subfloat[]{
			\includegraphics[scale=0.25]{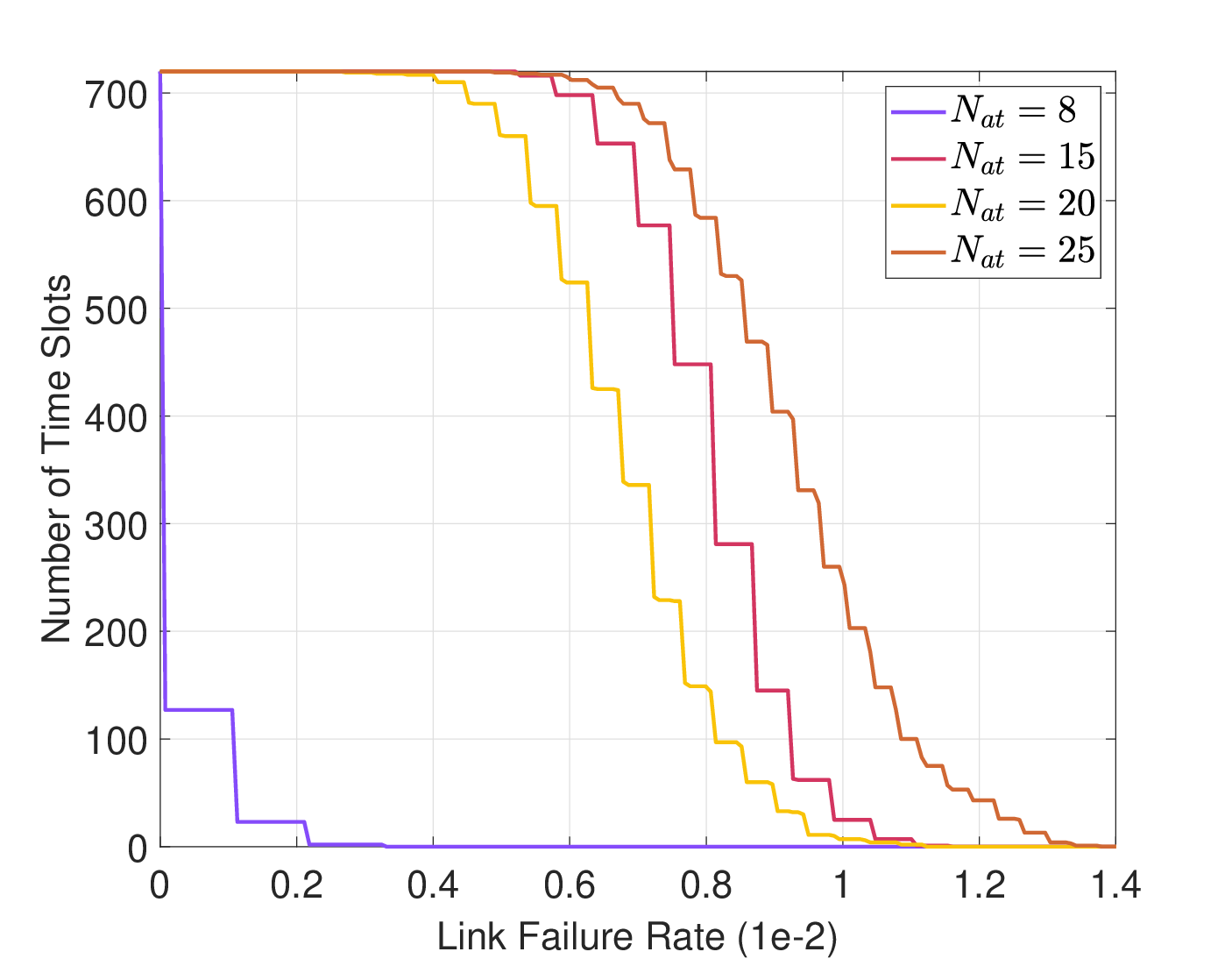}}
		\caption{Link failure rate performance. (a) Time distributions of LF rate for different IM methods, $N_{\rm at}=25$. (b) CCDF of aggregate I/N after IM for different IM methods, $N_{\rm at}=25$. (c) Time distributions of LF rate for different $N_{\rm at}$. }
		\label{lf}
		\vspace{-10pt}
	\end{figure*}
	
	\begin{figure*}[t]
		\centering
		\includegraphics[width=0.90\textwidth]{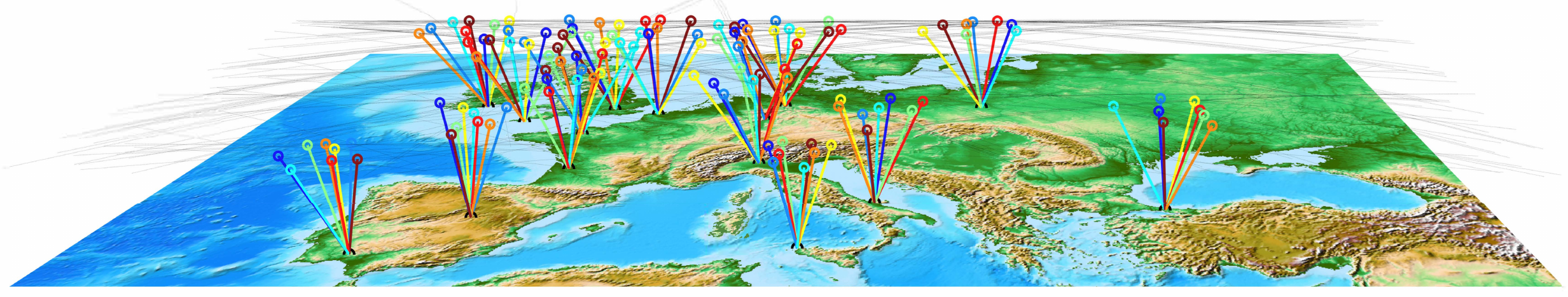}
		\caption{Visualization of subchannel allocation results obtained by GC-based IM.}
		\label{ksh}
	\end{figure*}

	\subsection{Simulation Setup}
	In this section, we {set up simulations} on the Starlink's first and second generation (SX1, SX2) mega LEO system to evaluate the performance of the proposed GC-based IM methodology.\footnote{For constellation coexistence, the proposed methodology is also effective under information sharing and cooperation between operators \cite{densesky}.} The system consists of 34396 satellites in 14 different orbital shells \cite{gen2starlink} and 115 MAGSs currently in operation  \cite{esstarlink}. Please refer to \cite{gen2starlink} for detailed orbital parameters. The ephemeris of the constellation is generated based on orbital parameters over a 2-hour period and sampled at an interval of 10 seconds, hence the number of time slots is $T_{\rm slot}=720$. Values of important system parameters are summarized in Table \ref{paras}. As illustrated in Section \ref{S-CTS}, the number of subchannels $C$ is set to $N_{\rm at}$ to achieve system-level IM, which will be validated later in Fig. \ref{M}. Without loss of generality, we also make worst-case assumption that all the satellites transmit signals at their maximum power.
	
	\renewcommand{\arraystretch}{1.3}
	\begin{table}[t]
		\centering 
		\small
		\caption{System Parameters}
		
		\begin{tabular}{ccc}
			\toprule
			\textbf{System parameters} & \textbf{Values} \\
			\midrule
			The number of satellites $S$ & 34396 \\
			The number of gateway stations $Q$ & 115 \\
			Satellite orbital altitude & 340-614km\\
			Peak gain of satellite antenna $G_{\rm s}(0)$ & 35dB \cite{s1528}\\
			Transmission power of satellite $P_{s}^{\rm tr}$ & 12dBW \cite{spacexant}\\
			Peak gain of gateway station antenna $G_{\rm g}(0)$ & 45.76dB \cite{s1428}\\
			Noise temperature of gateway station antenna $T_{\rm n}$ &  398K \cite{spacexant}\\
			Downlink carrier frequency $f_c$ & 20GHz \cite{schedules}\\
			Downlink bandwidth $B_{\rm w}$ & 500MHz \cite{schedules}\\
			Elevation angle threshold $\theta_{\rm thr}$ & $40^{\circ}$ \cite{spacexant}\\
			\bottomrule
		\end{tabular}
		\label{paras}
	\end{table}

	\begin{table}[t]
		\centering 
		\footnotesize
		\caption{Performance Comparison of Different Algorithms under $N_{\rm at}=25$.}
		
		\begin{tabular}{p{1.7cm}<{\centering}
				p{1.7cm}<{\centering}
				p{1.7cm}<{\centering}
				p{1.7cm}<{\centering}
			}
			\toprule	
			\textbf{Algorithms} & \textbf{Average LF rate (1e-2)} &  \textbf{Execution time (ms)} & \textbf{Memory (MB)}\\
			
			\midrule
			Random & 61.63 & $\approx 0$ & 0.02 \\ 
			Global & 6.14 & 2 & 4.4 \\ 
			HEAD & 1.20 & 445 & 47 \\ 
			Gurobi & 0.77 & $10^5$ & $8 \times 10^3$ \\
			GG & 1.16 & 6 & 5.4 \\ 
			CTS  & 0.89 & 20 & 4.7 \\ 
			
			\bottomrule
		\end{tabular}
		\label{algs}
	\end{table}
	
	In this section, we compare the following IM schemes:
	
	\begin{itemize}
		\item \textbf{Conventional IM techniques}: Three conventional IM techniques analyzed in \cite{densesky}. Random subchannel allocation (`\textbf{Random}') is the simplest frequency domain IM method, in which each feeder link uses a random subchannel to communicate. `\textbf{Look-aside}' alters satellite selection to avoid extremely small link angles. `\textbf{Power Control}' decreases satellite transmission power to reduce interference.
		\item `\textbf{Global}': \textit{Global} algorithm with $C$ colors. If conflicts occur for a vertex, it is assigned a random color.
		\item `\textbf{HEAD}': Hybrid Evolutionary Algorithm in Duet (HEAD) is one of the most powerful universal GC algorithms, outperforming others on the vast majority of benchmark graphs \cite{head}.
		\item `\textbf{Gurobi}': Gurobi is a renowned commercial solver which can solve the INLP problem P1 directly \cite{gurobi}.
		\item `\textbf{GG}' and `\textbf{TCFA-GG}': Generalized Global and its TCFA modification proposed in Section \ref{S-GG} and \ref{S-CFA}.
		\item `\textbf{CTS}' and `\textbf{TCFA-CTS}': Clique-Based Tabu Search and its TCFA modification proposed in Section \ref{S-CTS} and \ref{S-CFA}.
	\end{itemize}
	
	In addition, the algorithm parameters used in our simulations are: $I_{\rm th}^{\Gamma}=-13$dB for graph construction, $N_{\rm GG}=100$ for GG and \textit{Global}, and  $N^{t}_{\rm in}=2000$, $N^{t}_{\rm ca}=500$, $N_{\rm it}=250$, $N_{\rm n}=10$ for CTS. All simulations run on a computer with the 12th Gen Intel(R) Core(TM) i7-12700 processor and 32GB of RAM.
	
	\begin{figure*}
		\begin{minipage}{0.33\textwidth}
			\centering
			\includegraphics[width=\textwidth]{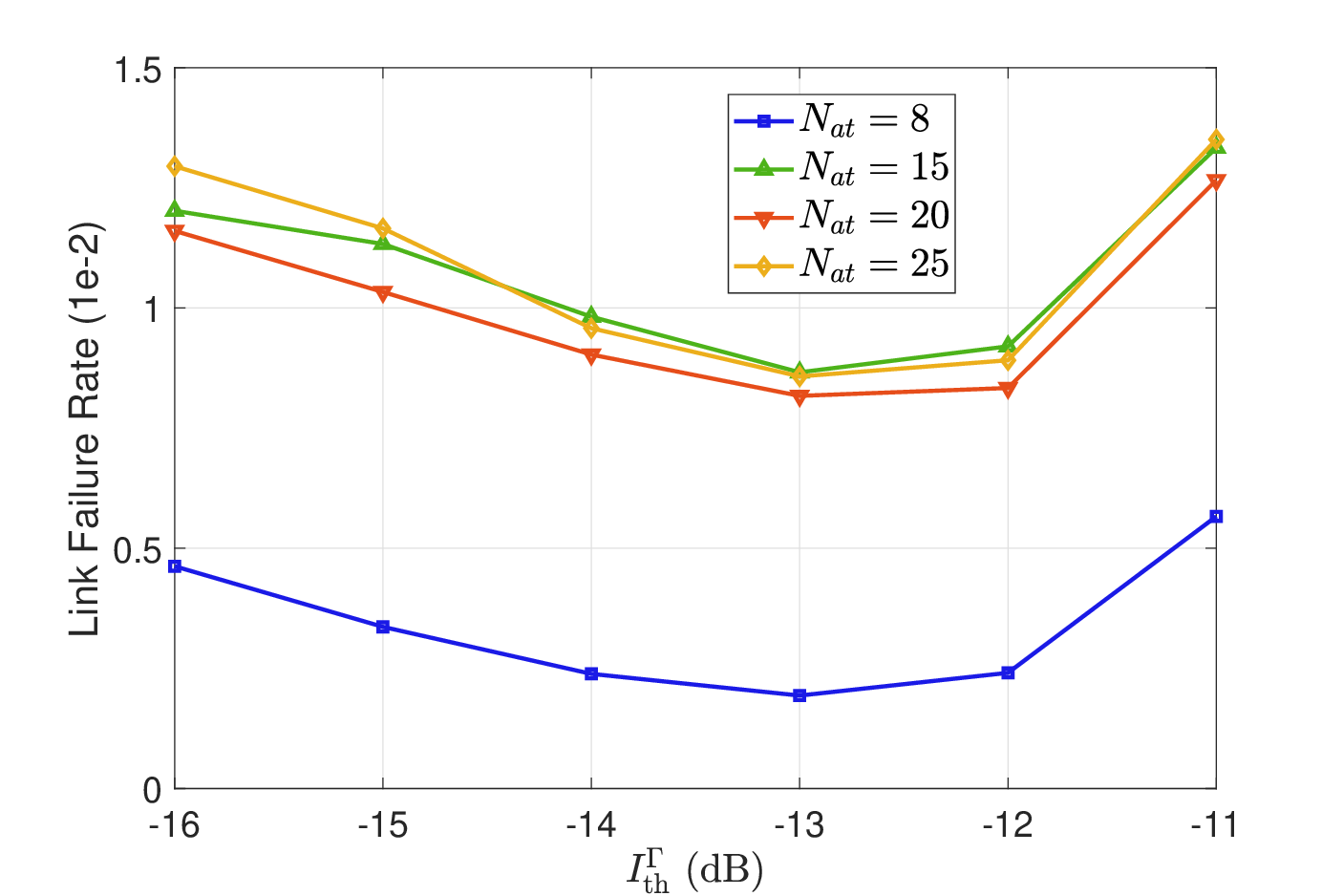}
			\caption{LF rate versus $I_{\rm th}^{\Gamma}$.}
			\label{ith}
		\end{minipage}
		\begin{minipage}{0.33\textwidth}
			\centering
			\includegraphics[width=\textwidth]{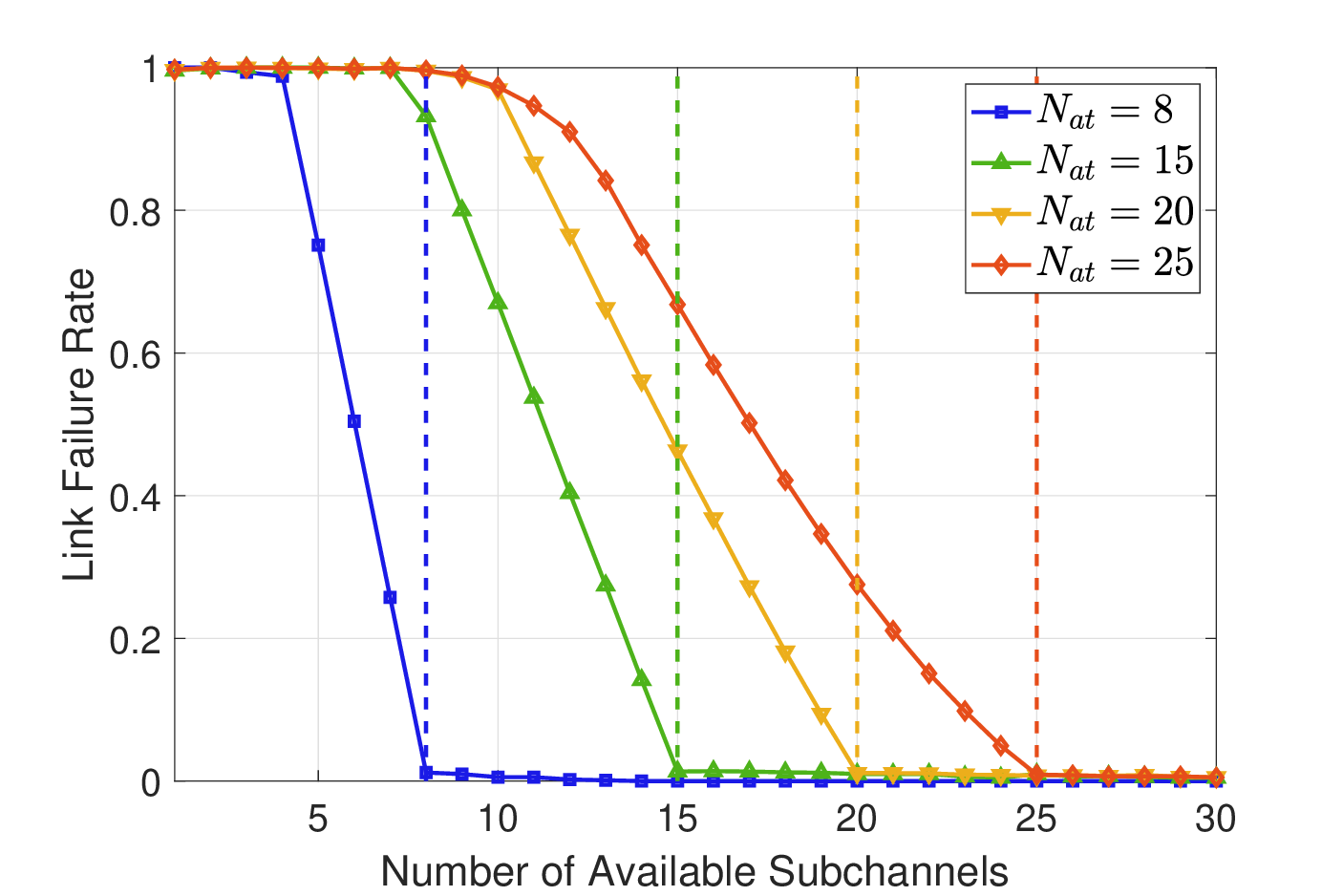}
			\caption{LF rate versus $C$.}
			\label{M}
		\end{minipage}
		\begin{minipage}{0.33\textwidth}
			\centering
			\includegraphics[width=\textwidth]{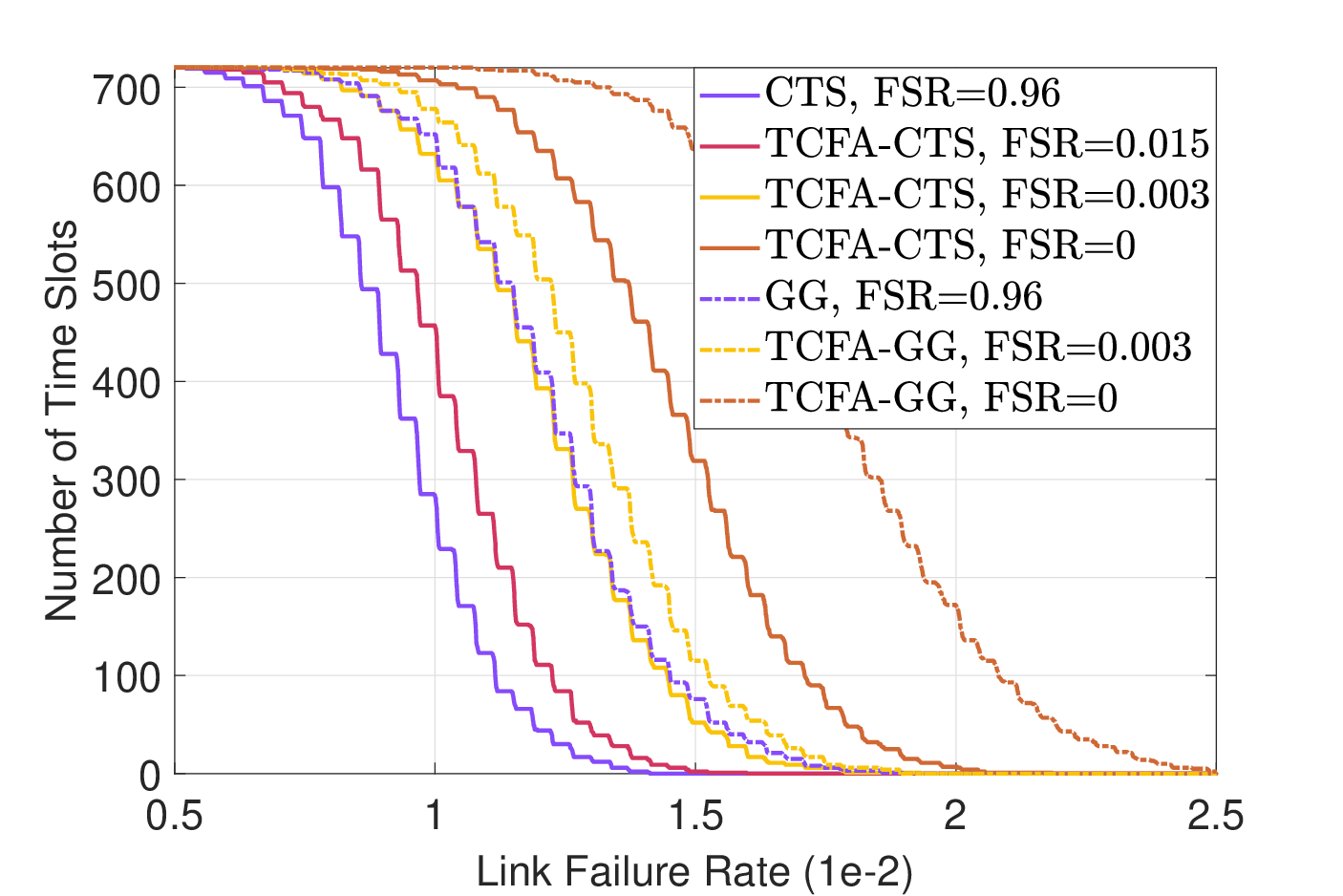}
			\caption{LF rate distribution under TCFA.}
			\label{fsr}
		\end{minipage}
		\vspace{-10pt}
	\end{figure*}
	
	\subsection{Link Failure Performance}
	Fig. \ref{lf}(a) and Table \ref{algs} show the LF rate performance of different frequency domain IM methods. For mega LEO systems with $N_{\rm at}$ up to 25 \cite{MAGSIC,constellations}, nearly 60\% of the feeder links will fail under random subchannel assignment. By introducing the GC-based IM methodology, the LF rate is substantially reduced to 6\% even with the simplest GC algorithm \textit{Global}. CTS, GG and HEAD all succeed in suppressing the average LF rate below 1\%, meaning that 99\% of the feeder links meet the requirement of ITU, so that IM of the entire system is almost achieved. In particular, CTS outperforms HEAD, one of the best universal GC algorithms, in both LF rate and time complexity due to effective utilization of the unique clique structure of $G(t)$, and GG delivers performance comparable to HEAD with remarkable computational efficiency. The modest performance gap between CTS and Gurobi arises from Gurobi's ability to incorporate the aggregate effect of numerous weak interference, albeit at exorbitant computational cost.
	
	Fig. \ref{lf}(b) provides the CCDF of aggregate I/N for various IM methods. Consistent with the results in \cite{densesky}, the superiority of frequency domain IM over space and power domain IM (look-aside and power control) in mega LEO systems is evident. 
	
	Fig. \ref{lf}(c) depicts the LF rate performance achieved by CTS under different $N_{\rm at}$. For current Starlink system with $N_{\rm at}=8$ \cite{starlinkiton}, perfect IM is achieved most of the time, and the LF rate is controlled below 1.5\% throughout the system development. Note that since the LF rate is normalized by $|\mathcal{W}(t)|$, it doesn't rise monotonously with the growth of $N_{\rm at}$.
	
	Fig. \ref{ksh} visualizes the subchannel allocation results obtained by the proposed methodology under $N_{\rm at} =8$ in Europe. Black semicircles represent GSs and lines with different colors stand for feeder links using different subchannels. 
	
	Moreover, the impact of different system and algorithm parameters on LF rate is further illustrated as follows.
	
	\textit{1) Impact of $I_{\rm th}^{\Gamma}$:} During the GC-based problem transformation, the predefined threshold $I_{\rm th}^{\Gamma}$ is a vital parameter that trades off between the attention of aggregate interference and the tractability of GC. Determining the appropriate value of $I_{\rm th}^{\Gamma}$ is called the threshold spectrum coloring (TSC) problem in \cite{tsc}, which can be solved through offline tuning before operation. Fig. \ref{ith} shows the performance of CTS under different $I_{\rm th}^{\Gamma}$. It can be observed that the optimal $I_{\rm th}^{\Gamma}$ is near -13 dB regardless of $N_{\rm at}$, which is very close to $I_{\rm th}^{\rm R}=$ -12.2dB. This indicates that reducing the number of edges to focus on strong interference is much more beneficial than considering the aggregate effect of weak interference. As a result, the gap between P1 and P2 is trivial, so that the proposed graph-based problem transformation is virtually optimal. 
	
	\textit{2) Impact of $C$:} Fig. \ref{M} validates the necessity and effectiveness of setting $C=N_{\rm at}$. We can observe that when $C<N_{\rm at}$, the LF rate increases almost linearly with $N_{\rm at}-C$ until 100\%. The LF rate is effectively controlled below 1\% under $C=N_{\rm at}$, and further increasing $C$ will only bring trivial benefit to LF rate while leading to plenty of vacant subchannels, thus severely degrading system capacity. Therefore, $C=N_{\rm at}$ achieves an excellent balance between accessibility and system capacity.
	
	\begin{figure*}[t]	
		\begin{minipage}{0.33\textwidth}
			\centering
			\includegraphics[width=\textwidth]{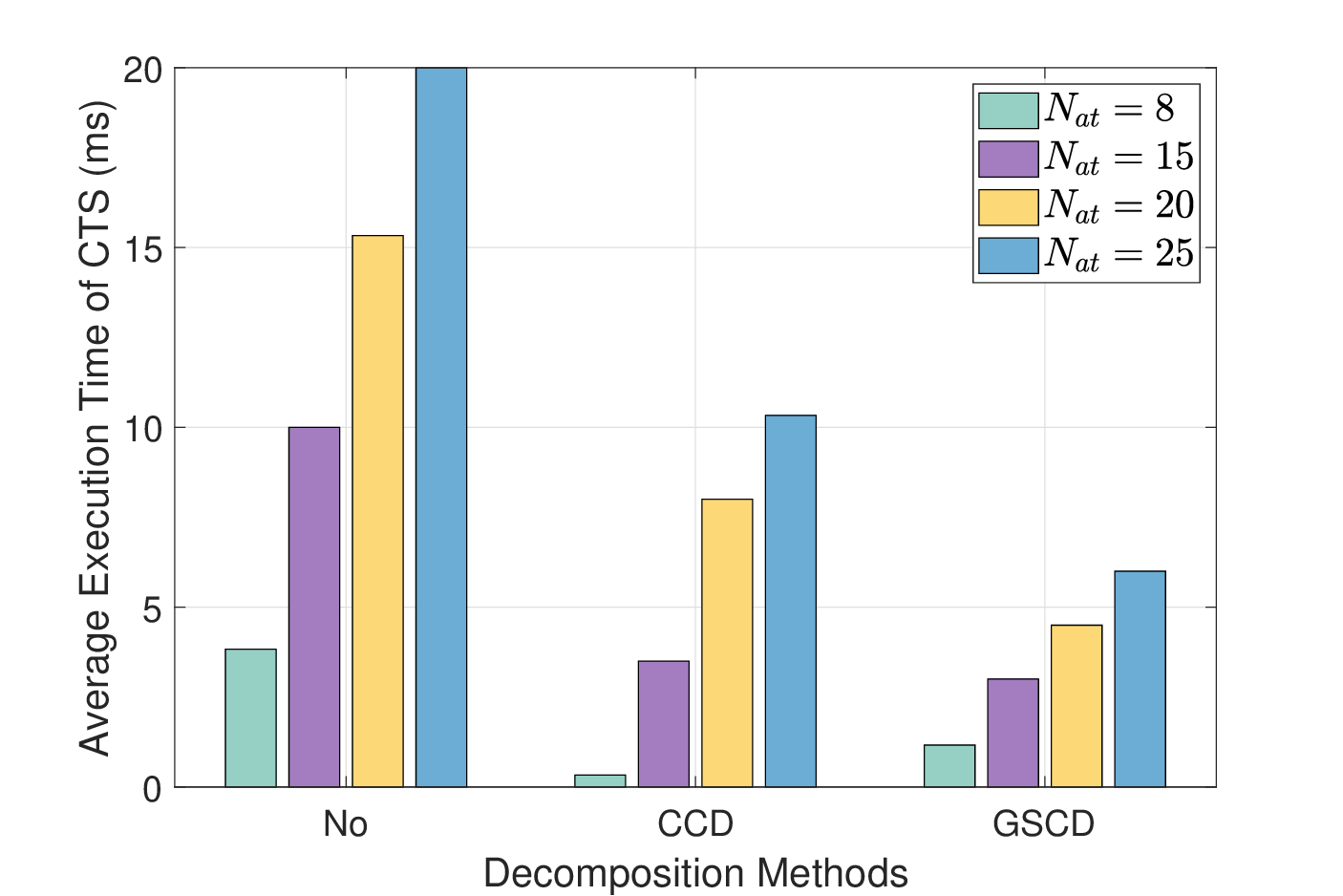}
			\caption{Performance of CCD and GSCD for CTS.}
			\label{dec}
		\end{minipage}
		\begin{minipage}{0.33\textwidth}
			\centering
			\includegraphics[width=\textwidth]{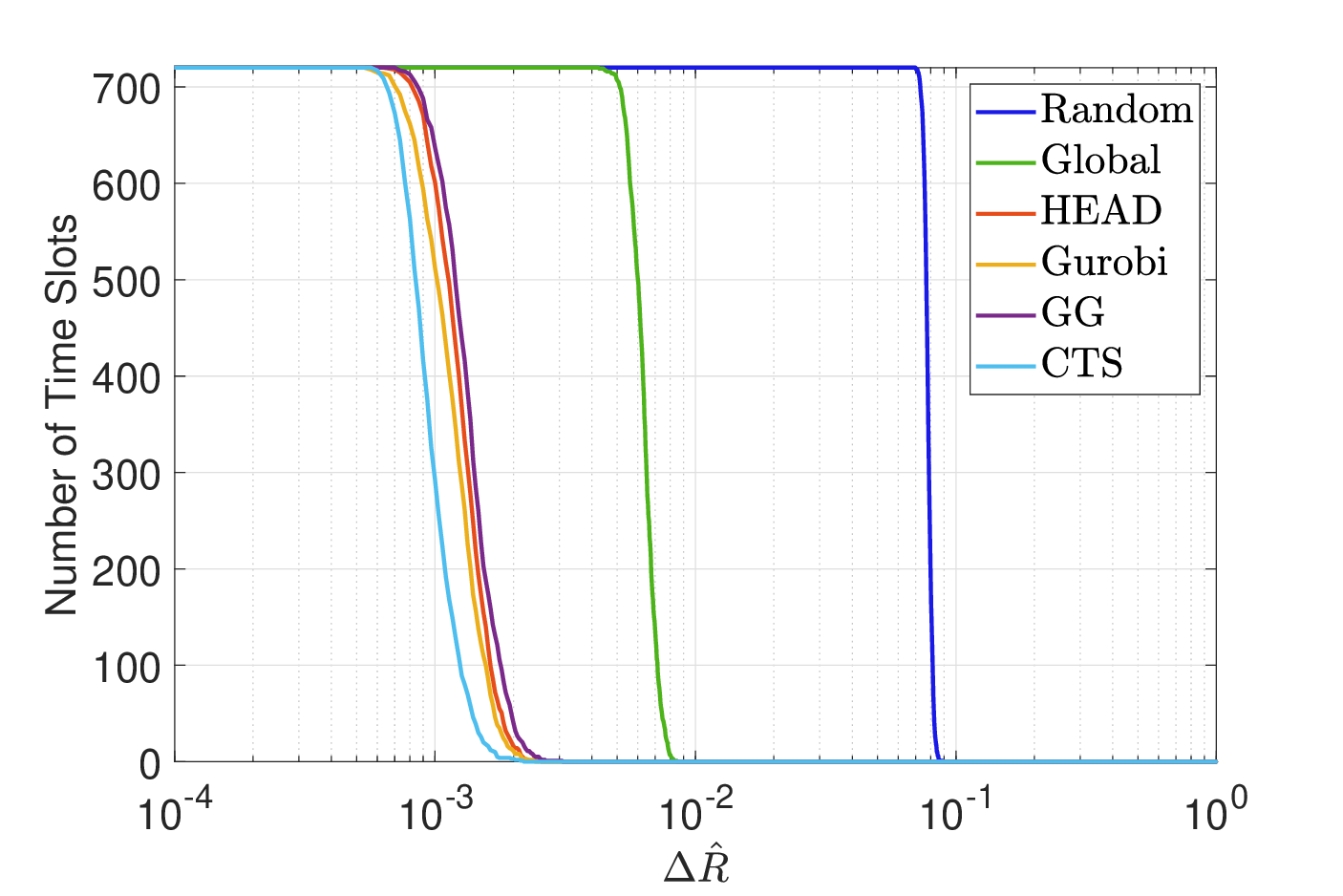}
			\caption{Time distribution of $\Delta \hat{R}$, $N_{\rm at}=25$.}
			\label{rdeg}
		\end{minipage}	
		\begin{minipage}{0.33\textwidth}
			\centering
			\includegraphics[width=\textwidth]{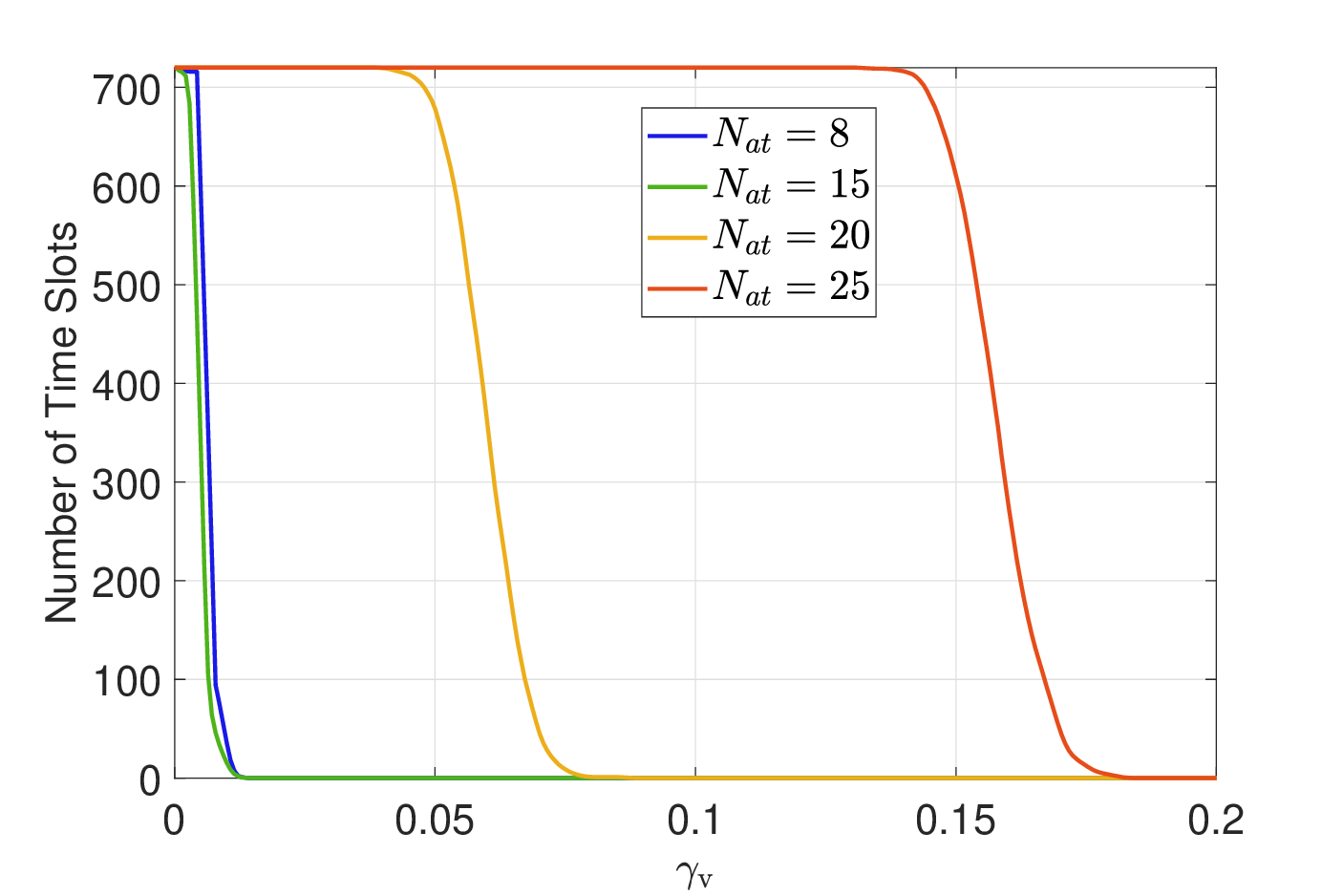}
			\caption{System capacity increase by VSU.}
			\label{td}
		\end{minipage}
		\vspace{-10pt}
	\end{figure*}
	
	\subsection{Performance of Time-Continuous Frequency Allocation}
	
	Fig. \ref{fsr} demonstrates the effectiveness of TCFA-CTS and TCFA-GG to achieve low LF rate and FSR simultaneously. The FSR during the simulation period is calculated by
	\begin{equation}
		{\rm FSR}=\frac{\sum\limits_{t=1}^{T_{\rm slot}} \sum\limits_{s \in \mathcal{S}}  \left(  1-\delta(c_s^t-c_s^{t-1}) \right) m_{\rm c}(s,t)}{\sum\limits_{t=1}^{T_{\rm slot}} \sum\limits_{s \in \mathcal{S}} m_{\rm c}(s,t)}.
	\end{equation}
	Simulation suggests that even under maximum elevation angle principle with the highest handover rate \cite{handover}, the probability of a satellite being selected by the same GS for two consecutive time slots still exceeds 70\% for $N_{\rm at}=25$ when the ephemeris is sampled at a 10-second resolution, emphasizing the necessity of considering TCFA. In this scenario, the FSR of all IM algorithms will exceed 95\% without addressing the issue of TCFA. In particular, solving the IM problem under TCFA constraints by Gurobi is computationally infeasible since the time and space complexity for single time slot is already overwhelming, and it further escalates exponentially when TCFA is considered. HEAD also fails to adapt to lower FSR since it requires unconstrained optimization over the entire graph. In comparison, the LF rate performance gaps between TCFA-CTS and CTS, as well as TCFA-GG and GG, are below 0.1\% even under low FSR tolerance. Particularly, the LF rate of TCFA-CTS is below 1.1\% under an FSR of only 0.015, demonstrating its great potential to achieve TCFA in mega LEO systems. Nevertheless, noticeable performance degradation can be observed when switching frequency is strictly prohibited since conflicts between constrained vertices are completely inevitable.
	
	\vspace{-2mm}
	\subsection{Effect of Mega Constellation Decomposition}
	
	Fig. \ref{dec} shows the performance of CCD and GSCD, where CTS algorithm is taken as example. It can be seen that the effect of CCD diminishes with the growth of $N_{\rm at}$, which can only save 50\% of the time under $N_{\rm at}=20$ and 25. In comparison, GSCD can stably reduce the complexity by nearly 70\% for all $N_{\rm at}$. Specially, CTS experiences almost no IM performance degradation after recoloring since structure \eqref{struct} is preserved throughout the entire process. By contrast, GSCD is not applicable to HEAD because the high performance of metaheuristic GC algorithms are based on global information of the entire graph. In this regard, the superiority of CTS is reinforced again.
	
	\subsection{Performance of System Capacity Maximization}
	
	Fig. \ref{rdeg} shows the distribution of $\Delta \hat{R}$ for different IM methods. It can be observed that except for random subchannel assignment which suffers from a  $\Delta \hat{R}$ of about 8\%, all IM methods based on GC achieve a $\Delta \hat{R}$ below 1\%. This confirms that the system capacity obtained by GC-based IM methods virtually approaches the theoretical upper bound. Notably, CTS outperforms other algorithms, including Gurobi, on system capacity, which succeeds in controlling the average $\Delta \hat{R}$ below 0.1\%. 
	
	Fig. \ref{td} depicts the effectiveness of VSU, which is evaluated by capacity improvement ratio $\gamma_{\rm v}=\frac{\sum_{s  \in \mathcal{W}(t)} R^{\rm r}_{s}}{\sum_{s \in \mathcal{W}(t)} {R}_{s}}-1$, where $R^{\rm r}_{s}$ denotes the link capacity of satellite $s$ after VSU. This technique is particularly useful for large $N_{\rm at}$ due to the emergence of a noticeable amount of vacant subchannels, as explained in Section \ref{S-uvs}. This can also be reflected in the density of subgraphs for MAGSs, which is denoted as $\rho(\mathscr{C}_i(t))$ for MAGS $i$.\footnote{$\rho(G)$ for general graph $G$ is calculated by $\frac{|E|}{|V|(|V|-1)}$, which peaks at 1 if $G$ is a complete graph (clique).} The average $\rho(\mathscr{C}_i(t))$ for $N_{\rm at}=8$ is 1 and decreases to 0.995, 0.944 and 0.823 for $N_{\rm at}=$ 15, 20, 25 respectively.
	VSU boosts system capacity by an average of 6\% and 17\% for $N_{\rm at}=$ 20 and 25 respectively, demonstrating its great potential for enhancing the spectrum efficiency of future mega LEO systems.
	
	\section{Conclusion}\label{SecE}
	
	This paper presents a time-continuous frequency allocation methodology based on graph coloring to address the challenges of strong interference in mega LEO systems with multi-antenna gateway stations. We utilize the unique characteristics of MAGSs to design GC algorithms with higher performance, and further modify them to achieve system-level IM and TCFA simultaneously. Additionally, we devise a list coloring-based vacant subchannel utilization method to further enhance spectrum efficiency and system capacity. Our research introduces the GC methodology to solve the IM problem of mega LEO systems with MAGSs for the first time, demonstrating substantial improvements on various performance metrics. {Future research directions include developing protocols for applying the proposed method in real mega LEO systems, as well as addressing the fairness issue of IM in constellation coexistence scenarios.}
	
	\appendices
	
	\section{Proof of Equation \eqref{dr}}\label{Apenproof}
	To prove $\delta_R>0$, we first introduce ratio $\gamma_R$ and rewrite inequality \eqref{21} as
	\begin{equation}
		(C/N)_{s} = \gamma_R {\rm SINR}_{s}, 1 \leq \gamma_R \leq 1+I_{\rm th}^{\rm R}.
	\end{equation}
	Since $\gamma_R>1$, we can immediately find that
	\begin{equation}
		\frac{\log_2(1+{\rm SINR}_{s})}{\log_2(1+(C/N)_{s})} > \frac{\log_2(1+{\rm SINR}_{s})}{\log_2(\gamma_R+\gamma_R {\rm SINR}_{s})}.
		\label{fangsuo}
	\end{equation}
	Then, we construct function $f(x)$ as follows
	\begin{equation}
		f(x)=\frac{\ln x}{\ln \gamma_Rx}=1-\frac{\ln \gamma_R}{\ln \gamma_R+\ln x},
	\end{equation}
	based on which the following equation holds:
	\begin{equation}
		\begin{aligned}
			f(x+1)-f(x) 
			=\frac{\ln \gamma_R \ln\frac{x+1}{x} }{(\ln \gamma_R +\ln x)[\ln \gamma_R +\ln (x+1)]}.
		\end{aligned}
	\end{equation}
	When $x={\rm SINR}_{s}$ and ${\rm SINR}_{s}>1$, the value of $f({\rm SINR}_{s}+1)-f({\rm SINR}_{s})$ is positive. Considering \eqref{fangsuo}, we finally get
	\begin{equation}
		\delta_R>f({\rm SINR}_{s}+1)-f({\rm SINR}_{s})>0.
	\end{equation}
	This completes the proof.

	\bibliographystyle{IEEEtran}
	\bibliography{reference}

\end{document}